\documentclass[10pt]{article}
\usepackage[margin=0.80in]{geometry}
\usepackage{amsmath,amssymb,amsthm,mathtools}
\usepackage{microtype}
\usepackage{enumitem}
\usepackage{xcolor}
\usepackage[colorlinks=true,linkcolor=blue!45!black,citecolor=blue!45!black,urlcolor=blue!45!black]{hyperref}
\usepackage[capitalize,nameinlink]{cleveref}
\usepackage[utf8]{inputenc}
\usepackage[T1]{fontenc}
\usepackage{csquotes}
\usepackage{upgreek}

\usepackage[backend=biber, style=numeric-comp,sorting=none]{biblatex}

\usepackage{algorithm,algpseudocode}

\newtheorem{theorem}{Theorem}[section]
\newtheorem{proposition}[theorem]{Proposition}

\newtheorem{lemma}[theorem]{Lemma}
\newtheorem{corollary}[theorem]{Corollary}
\newtheorem{assumption}{Assumption}

\theoremstyle{definition}
\newtheorem{remark}[theorem]{Remark}
\newtheorem*{remark*}{Remark}

\newcommand{\Rbb}{\mathbb R}
\newcommand{\Cbb}{\mathbb C}
\newcommand{\Nbb}{\mathbb N}
\newcommand{\Ebb}{\mathbb E}
\newcommand{\CC}{\mathbb C}

\newcommand{\cA}{\mathcal A}
\newcommand{\cE}{\mathcal E}
\newcommand{\cL}{\mathcal L}
\newcommand{\cP}{\mathcal P}
\newcommand{\cN}{\mathcal N}
\newcommand{\cG}{\mathcal G}

\newcommand{\cD}{\mathcal D}

\newcommand{\cB}{\mathcal B}
\newcommand{\cK}{\mathcal K}
\newcommand{\cR}{\mathcal R}
\newcommand{\Trunc}{\mathsf T}

\newcommand{\supp}{\mathrm{supp}}
\newcommand{\ci}{\mathrm i}
\newcommand{\eps}{\varepsilon}
\newcommand{\poly}{\mathrm{poly}}
\newcommand{\Fone}{1}
\newcommand{\dd}{\mathrm d}
\DeclareMathOperator{\arsinh}{arsinh}

\DeclareMathOperator{\sech}{sech}
\DeclareMathOperator{\tr}{tr}
\DeclareMathOperator{\Tr}{Tr}
\newcommand{\abs}[1]{\left| #1 \right|}
\newcommand{\tmax}{t_\mathrm{max}}
\newcommand{\row}{\mathrm{row}}
\newcommand{\norm}[1]{\| #1 \|}

\newcommand{\be}{\begin{equation}}
\newcommand{\ee}{\end{equation}}
\newcommand{\F}{\mathrm{F}}
\newcommand{\HSproduct}[2]{\langle #1 , #2 \rangle}
\renewcommand{\gamma}{\upgamma}
\newcommand{\<}{\left\langle}
\renewcommand{\>}{\right\rangle}
\renewcommand{\P}{\mathsf{P}}

\definecolor{note-color}{RGB}{227,116,0}

\newcommand{\ZZ}{\mathbb{Z}}
\newcommand{\RR}{\mathbb{R}}
\newcommand{\NN}{\mathbb{N}}

\newcommand{\BB}{\cB}

 \newcommand{\g}{\text{g}}

\title{Fast convergence of propagation algorithms for open quantum systems}

\usepackage{authblk}
\author[1]{Giorgio~Facelli}
\author[1]{Hamza~Fawzi}
\author[2]{Omar~Fawzi}
\author[3]{Cambyse Rouz\'{e}}
\author[2]{Sam Slezak}
\author[4]{Daniel Stilck Fran\c{c}a}
\affil[1]{Department of Applied Mathematics and Theoretical Physics, University of Cambridge, United Kingdom}
\affil[2]{Inria, ENS Lyon, UCBL, LIP, F-69342 Lyon Cedex 07, France}
\affil[3]{Inria, Télécom Paris - LTCI, Institut Polytechnique de Paris, France}
\affil[4]{Department of Mathematical Sciences, University of Copenhagen, Denmark}

\date{August 2026}

\begin{document}
\maketitle

\begin{abstract}
We establish general conditions for the rapid convergence of propagation algorithms for computing expectation values of local observables in quantum many-body systems, covering both spin and fermionic systems. Our results apply to evolutions that are sufficiently contractive, providing a general mechanism by which contractivity controls the growth generated by local interactions and enables efficient classical simulation.

As a first application, we consider noisy time evolution generated by local Hamiltonians on arbitrary interaction graphs. We show that the dynamics can be efficiently simulated when the noise strength $\lambda$ is sufficiently large compared to the degree of the interaction graph.
More generally, for systems with an interaction strength $u$, our bounds determine a time horizon $t_{\text{max}}$ as a function of the interaction strength $u$, the noise strength $\lambda$, and the interaction structure below which the algorithm is efficient.
In the noiseless limit, $\lambda=0$, 
our analysis extends the efficiently simulable time scale from $\log(1/u)$ proved in Facelli, Fawzi, and Fawzi (2026) to $t_{\text{max}} \sim 1/u$ matching the recent improvement by Zhao, Marvian and Tong (2026).

As a second application, we apply the same framework to the computation of local observables in Gibbs states of local Hamiltonians. We consider a pseudo-Lindbladian whose steady state is the Gibbs state and show that, at sufficiently high temperature, its propagation algorithm converges rapidly.  In particular, for weakly-interacting systems with interaction strength $u$, our result applies up to inverse temperatures $\beta\sim\log(1/u)$.
\end{abstract}

\section{Introduction}

Understanding the dynamics and equilibrium properties of quantum many-body systems is a central challenge in quantum physics. The simulation of closed-system Hamiltonian dynamics has been a major focus of quantum computation since the foundational work of Feynman and Lloyd~\cite{feynman1982simulating,lloyd1996universal}. Here, we focus on open quantum systems, where dissipation plays a central role in the dynamics and the approach to equilibrium.

Efficient quantum algorithms have been developed for simulating open-system dynamics over a prescribed time interval~\cite{childs2016efficient,cleve2016efficient,li2022simulating}. More recently, quantum Gibbs samplers have been constructed by designing efficiently implementable dissipative dynamics whose stationary state is the desired Gibbs state~\cite{chen2025efficient,ding2025efficient}. Efficient Gibbs-state preparation additionally requires these dynamics to converge rapidly, a property established for several classes of Hamiltonians and temperature regimes~\cite{rouze2026efficient,rouze2026optimal,ramkumar2024mixing,smid2025polynomial,tong2025fast,bergamaschi2026fastmixingquantumspin,bergamaschi2026fastmixingalltoallquantum}.

Classical methods for computing thermal properties have a long history, including quantum Monte Carlo~\cite{blankenbecler1981monte,sandvik1999stochastic} and finite-temperature tensor-network methods~\cite{verstraete2004matrix,feiguin2005finite}. Rigorous algorithmic results include the efficient approximation of one-dimensional Gibbs states~\cite{kuwahara2021improved} and their local thermal observables~\cite{alhambra2021locally}, as well as algorithms with provable guarantees for computing thermal quantities~\cite{fawzi2023subpolynomial,fawzi2024certified,zlokapa2026rigorous}, approximating partition functions~\cite{harrow2020classical,chen2025convergence,crosson2025classical,mann2026efficient}, and sampling from the measurement distributions of Gibbs states~\cite{bakshi2024hightemperature,yin2023polynomialtime,ramkumar2026high}, under various structural and temperature assumptions. These equilibrium calculations need not simulate the underlying relaxation dynamics. Establishing the scope of classical algorithms for both dynamics and equilibrium properties is therefore essential for studying quantum many-body systems and identifying where quantum computers may offer a computational advantage.

A natural approach to classical simulation works in the Heisenberg picture, expanding observables in a basis of Pauli or Majorana strings. \emph{Propagation algorithms}, developed for approximating quantum-circuit outputs~\cite{rall2019simulation,aharonov2023polynomial,beguvsic2025simulating,schuster2025polynomial,angrisani2025classically,angrisani2026simulating,miller2025simulation}, evolve these expansions by alternating short steps of the dynamics with truncations that keep the representation computationally manageable. Despite their practical success~\cite{begusic2024fast,rudolph2023classical,begusic2025realtime,danna2026majorana}, obtaining rigorous bounds on the accumulated truncation error remains challenging. In particular, coherent interactions can generate strings of increasing degree, enlarging the space of operators needed to represent the evolved observable. In open systems, dissipation can counteract this growth by suppressing high-degree components. Understanding when this competition permits efficient classical simulation requires quantitative bounds on both the complexity of the retained expansion and the error introduced by truncation.

For Hamiltonian dynamics, Ref.~\cite{facelli2026fast}, by a subset of the authors, established that propagation algorithms can efficiently construct accurate approximations whenever the evolved observable remains well approximated by low-degree operators throughout the evolution. It further showed that this condition holds for weakly interacting Hamiltonians over a time window that grows logarithmically with the inverse interaction strength. These results connect the existence of low-degree approximations to the performance of an explicit classical algorithm.

In this work, we develop a contraction-based analysis of propagation algorithms for open quantum systems, with particular emphasis on weakly interacting Hamiltonians. For Hamiltonian dynamics subject to noise, we quantify how interaction strength and dissipation determine the time interval over which low-degree propagation remains accurate. We also construct an evolution generated by a pseudo-Lindbladian that allows the same approach to compute local thermal expectation values in weakly interacting systems.

\subsection{Main results}

For concreteness, we present our main results for fermionic systems with $N$ Majorana operators $\gamma_1,\ldots,\gamma_N$. The general framework, which also covers spin systems, is developed in Section~\ref{sec:norms-shared-setup} and the subsequent analysis. Every observable admits an expansion
\[
A=\sum_{S\subseteq[N]} a_S\gamma_S,
\]
where $\gamma_S$ denotes the self-adjoint Majorana monomial associated with $S$, and its degree is $|S|$. We write $\Trunc_\ell$ for the truncation that discards all monomials of degree greater than $\ell$. The resulting expansion contains at most $\sum_{k=0}^{\ell}\binom{N}{k}$ coefficients and can therefore be stored efficiently when $\ell$ is small.

Given an evolution $\cP_t=e^{t\cL}$ and an efficiently computable short-time approximation $\widetilde{\cP}_h$, a propagation algorithm approximates $\cP_t A$ by iterating
\begin{equation}
\label{eq:intro-propagation-algo}
A_0=A,\qquad
A_{j+1}=\Trunc_\ell\widetilde{\cP}_h(A_j).
\end{equation}
Our goal is to identify conditions under which a low-degree cutoff suffices to control the accumulated error, and to quantify the resulting computational cost.

\paragraph{Weighted contraction}
The central ingredient of our analysis is contraction in norms that assign exponentially larger weights to higher-degree monomials. Define the degree operator by $\cN(\gamma_S)=|S|\gamma_S$, so that $e^{a\cN}$ multiplies the coefficient of a degree-$k$ monomial by $e^{ak}$. In its simplest form, the contraction estimate reads
\begin{equation}
\label{eq:intro-weighted-contraction}
\|e^{a\cN}\cP_t B\|
\leq
\|e^{a\cN}B\|,
\end{equation}
for some $a>0$ and a suitable norm of the Majorana coefficients. Because high-degree components carry a large weight, this estimate forces their unweighted contribution to be small. More precisely, for an initial observable $A$ of degree at most $d$ and a cutoff $\ell\geq d$, it gives
\[
\|(I-\Trunc_\ell)\cP_t A\|
\leq
e^{-a\ell}\|e^{a\cN}\cP_t A\|
\leq
e^{-a(\ell-d)}\|A\|.
\]
Thus the truncation error decreases exponentially with the cutoff. Combining this tail bound with stability estimates and a bound on the short-time approximation error yields rigorous guarantees for the full propagation algorithm. We establish the required contraction using the Frobenius norm for dissipative Hamiltonian dynamics and weighted coefficient $\ell_1$ norms for thermal expectation values. When dissipation is weak or absent, a time-dependent exponential weight allows us to control degree growth over a finite time interval.

\paragraph{Weakly interacting dynamics under dissipation.}
Our first application concerns the evolution
\begin{equation}
\label{eq:dissipative-evolution}
A(t)=e^{t\cL}A,
\qquad
\cL(A)=\ci[H_0+uV,A]-\lambda\cN(A),
\end{equation}
where $H_0$ is quadratic, $V$ is a quartic interaction with bounded coefficients, and $u,\lambda\geq0$ control the interaction and dissipation strengths. We assume that $H=H_0+uV$ is $\Delta$-sparse: each Majorana mode occurs in at most $\Delta$ interaction terms. No geometric locality assumption is imposed on the interaction graph.

The three terms in the generator have distinct effects on the Majorana expansion. Evolution under the quadratic Hamiltonian $H_0$ preserves degree, the interaction $uV$ can increase it, and the noise suppresses a degree-$k$ monomial at rate $\lambda k$. Our contraction estimate quantifies this competition. When $\lambda$ is sufficiently large compared with $\Delta u$, a fixed exponential weight gives a truncation-tail bound uniform in time. For weaker dissipation, allowing the weight to decrease with time yields a finite simulation horizon. In particular, in the noiseless limit this horizon scales as $1/(\Delta u)$, extending the logarithmic dependence on $1/u$ established in Ref.~\cite{facelli2026fast} and matching the scaling in Ref.~\cite{zhao2026efficient}.

\begin{theorem}[Simulation of fermionic dissipative dynamics, informal version of Corollary~\ref{cor:weak-dissipative-dynamics-fermionic}]
Consider a system of $N$ Majorana modes with $\Delta$-sparse Hamiltonian $H=H_0+uV$ where $H_0$ is quadratic and $V$ is quartic, where $u \leq 1$. Consider the evolution defined by Eq.~\eqref{eq:dissipative-evolution} where $A$ is an observable of fixed degree $d$. Define
\[
\lambda_c(u)=2\Delta u\sinh(1)
\]
and
\[
\tmax(u,\lambda)=
\begin{cases}
\displaystyle
\frac{1}{8\Delta u\cosh(1)}
\log\!\left(1+\frac{\Delta u}{\lambda_c(u)-\lambda}\right),
& \lambda<\lambda_c(u),\\[1.2ex]
\infty,
& \lambda\geq\lambda_c(u).
\end{cases}
\]
For any finite $0\leq t\leq\tmax(u,\lambda)$ and $0<\eps<1$, the propagation algorithm produces a degree-$\ell$ approximation $\widetilde A$ satisfying
\[
\|A(t)-\widetilde A\|_{\F}
\leq \eps\|A\|_{\F},
\qquad
\ell=d+O\!\left(\log\frac{1+t}{\eps}\right),
\]
in time $N^{O(\log((1+t)/\eps))}$, where $\|\cdot\|_{\F}$ is the normalized Frobenius norm.
\end{theorem}

\paragraph{Thermal expectations via a contractive evolution.}
Our second application concerns the computation of
\[
\Tr(A\rho_{\beta,H}),
\qquad
\rho_{\beta,H}=\frac{e^{-\beta H}}{\Tr(e^{-\beta H})},
\qquad
H=H_0+uV.
\]
We approach this equilibrium problem through an auxiliary evolution generated by a pseudo-Lindbladian $\cK^\beta$, inspired by Refs.~\cite{bergamaschi2026fastmixingquantumspin,bergamaschi2026fastmixingalltoallquantum}. Its semigroup $\cP_t=e^{t\cK^\beta}$ is unital, and its adjoint fixes $\rho_{\beta,H}$. Consequently, if $\cP_t A$ converges to a scalar multiple of the identity, that scalar must be $\Tr(A\rho_{\beta,H})$. The algorithm estimates this scalar by propagating $A$ and extracting the coefficient of the identity. The auxiliary evolution need not be completely positive; its role is to provide a computational route to the desired expectation value.

Here, contraction serves two purposes: it controls the high-degree tail and establishes exponential decay of the nonidentity component, proving that
\[
\cP_t A\longrightarrow \Tr(A\rho_{\beta,H})I.
\]
To obtain this contraction, we decompose $\cK^\beta$ into a reference generator associated with the quadratic Hamiltonian $H_0$ and a perturbation due to $uV$. The reference evolution contracts a suitable exponentially weighted coefficient $\ell_1$ norm. A Dyson expansion bounds the perturbation, showing that contraction persists when the interaction is sufficiently weak relative to the temperature. Together, these estimates yield a deterministic propagation algorithm for thermal expectation values.

\begin{theorem}[Approximation of fermionic thermal expectations, informal version of Theorem~\ref{thm:gibbs-observables}]
\label{thm:intro-gibbs}
Consider a system of $N$ Majorana modes with a $\Delta$-sparse Hamiltonian $H=H_0+uV$, where $V$ is a quartic interaction, and $H_0$ is a quadratic Hamiltonian geometrically local on a $D$-dimensional lattice. Let $A$ be a Majorana monomial of even degree $d$.\footnote{The restriction to even observables is for convenience only.}
There is an absolute constant $c>0$ and a constant $C>0$, depending only on $D$, such that whenever
\[
u \Delta e^{C\beta}(e^{2\beta}-1)\leq c,
\]
a deterministic propagation algorithm returns $\widetilde\mu$ satisfying
\[
\left|\widetilde\mu-\Tr(A\rho_{\beta,H})\right|
\leq\eps\|A\|_\infty
\]
in time
\[
(N+1)^{O(d+\log(1/\eps))}\poly(d,1/\eps).
\]
\end{theorem}

For fixed $\Delta$, the condition in Theorem~\ref{thm:intro-gibbs} permits inverse temperatures of order $\log(1/u)$ as $u\to0$. Thus weak interactions allow the propagation algorithm to extend beyond a fixed high-temperature regime. 

Recent work has established rigorous quantum and classical algorithms for thermal computation in weakly interacting systems, at inverse temperatures that grow as the interaction strength decreases~\cite{smid2025polynomial,tong2025fast,mann2026efficient,chen2025convergence}. Using a cluster expansion, Ref.~\cite{mann2026efficient} gives a deterministic approximation scheme for partition functions with runtime $\poly(N,1/\eps)$ at fixed $\beta$ and bounded interaction degree and rank, permitting inverse temperatures of order \(\log(1/u)\). However, this algorithm requires the noninteracting part $H_0$ to be on-site, i.e., it does not allow hopping terms of the form $\gamma_{i} \gamma_{j}$ for unpaired $i,j$. Allowing hopping terms is nontrivial as diagonalizing $H_0$ generally requires a nonlocal change of basis, which can destroy the locality of the interaction term $V$.

Ref.~\cite{chen2025convergence} handles quadratic Hamiltonians with hopping through a cumulant expansion combined with tree-based sampling. It gives randomized algorithms for the log-partition function and local thermal observables with runtime $\poly(N,1/\eps)$ at fixed $\beta$. For a $D$-dimensional lattice, its sufficient condition allows $\beta \lesssim u^{-1/(D+2)}$. More generally this algorithm only requires summability of the interactions and the noninteracting Green's function.

Although our runtime guarantees are weaker than those of the algorithm based on cumulant expansion, we obtain a simple, deterministic propagation algorithm that accommodates hopping in $H_0$ and requires no geometric locality of $V$. The lattice assumption in Theorem~1.2 can also be replaced by summability conditions similar to the summability of the noninteracting Green's function from~\cite{chen2025convergence} for bounded-degree graphs (see Assumption \ref{assumption:diagonal-hamiltonian-fermion}). We believe that the runtime guarantees we obtain can be improved by analyzing more selective truncation rules that exploit coefficient magnitudes or spatial structure. Establishing guarantees for such refinements is a natural direction for future work.

\section{Setup}
\label{sec:algorithm}

\subsection{Algebra and monomial basis}
\label{sec:norms-shared-setup}

To treat the Pauli and Majorana case simultaneously, we consider a finite-dimensional $C^*$-algebra $\mathcal A$ spanned by a basis of
monomials $\BB$ equipped with a notion of degree. Throughout, $N$ denotes the maximal
degree (the number of qubits or modes), and $\tr$ is the normalized trace on $\mathcal A$, so
$\tr(I)=1$. For $r\in\{0,\dots,N\}$ set
\be
\label{wth:eq:degree}
  \cA_r=\operatorname{span}\{P\in\BB:\deg(P)=r\},\qquad
  \cA_{\le\ell}=\bigoplus_{r\le\ell}\cA_r
\ee
For a nonzero $X$, $\deg(X)$ is the largest degree of a monomial with nonzero coefficient
in $X$. The projection $\Trunc_\ell$ keeps only the monomials of degree at
most $\ell$, i.e.\ it is the projection onto $\cA_{\le\ell}$.

We assume $\BB$ satisfies the following: $\BB \cap \cA_0=\{I\}$, and for all $P,Q\in\BB$,
\begin{enumerate}[label=(\roman*)]
  \item $P^2=I$ and $P^*=P$, and $\tr(P)=0$ whenever $\deg(P)>0$;
  \item $\deg(PQ)\le\deg(P)+\deg(Q)$;
  \item either $[P,Q]=0$, or $[P,Q]=2PQ=2f_{P,Q}S_{P,Q}$ for some $S_{P,Q}\in\BB$ and $f_{P,Q} \in \CC$ with $|f_{P,Q}|=1$;
  \item $|\{P\in\BB:\deg(P)\le r\}|\le (N+1)^{O(r)}$ for every $r\in\{0,\dots,N\}$.
\end{enumerate}
\medskip

\noindent The main examples of such bases are described below.
\paragraph{Pauli basis}
Consider the algebra $\cA$ of spins with system size $N$. The basis of the algebra is given by
\be
\label{eq:pauli-algebra}
    \BB = \biggr\{ P\in(\Cbb^{2\times 2})^{\otimes N}: P = P_1 \otimes \dots \otimes P_N \text{ s.t. } P_j\in \{I,X,Y,Z\} \biggr\}
\ee
where $\{I,X,Y,Z\}$ denote the usual $2\times 2$ Pauli operators along with the identity. For any string $P\in\BB$, its degree is then the number of sites where $P$ acts non-trivially, i.e.
\[
    \deg(P) =\abs{\supp(P)}\,, \qquad \supp(P) := \{i\in [N]: P_i \neq I\}\,,
\]
where $P_i$ above is intended as the $i$-th tensor factor of $P$ following the notation in \eqref{eq:pauli-algebra}. 

\paragraph{Majorana basis}
Consider now the case of a fermionic algebra $\cA$ on $N/2$ sites, where $N$ is assumed even. Let $\{\gamma_1,\dots,\gamma_N\}$ denote the corresponding self-adjoint Majorana operators satisfying the following canonical anticommutation relations
\be\label{eq:majorana-operators}
    \{\gamma_i,\gamma_j\} =2\delta_{ij} \,,\qquad i,j\in\{1,\dots,N\}.
\ee
We consider the operator basis of self-adjoint Majorana strings defined via the Majorana operators. Specifically, let
\be
\label{eq:majorana-basis}
    \BB := \{\gamma_X\}_{X\subseteq[N]} \,, \quad \gamma_X = \ci^{r_X} \prod_{j\in X} \gamma_j \,, 
\ee
where $r_X = \abs{X}(\abs{X}-1)/2 \mod 2$ ensures that $\gamma_X^* = \gamma_X$. For any $X\subseteq[N]$, we then define $\deg(\gamma_X) \equiv \abs{X}$.

\subsection{The degree operator and depolarising noise}
\label{sec:degree-operator}

We define the \emph{degree operator} $\cN : \cA \to \cA$ by
\begin{equation}
\label{eq:degree-operator}
\cN(P)=\deg(P)\,P \qquad \forall P \in \BB.
\end{equation}
The semigroup $e^{-t\cN}$ generated by $-\cN$ is the quantum depolarising noise, a.k.a. the quantum Ornstein-Uhlenbeck semigroup. For spin systems, it is given by
\be
\label{eq:depolarizing-noise-definition}
    \cD(X) := \frac{1}{4} \sum_{i\in [N]}\sum_{\alpha\in \{x,y,z\}} (P_{i,\alpha}XP_{i,\alpha} - X) \,.
\ee
For any Pauli string $P\in\BB$ , it acts as $\cD(P)= -\deg(P) P = -\cN(P)$.

In the fermionic case, let $\Pi = \ci^{N/2} \gamma_1\cdots\gamma_N$ be the \emph{fermion parity} operator, a self-adjoint unitary operator. It takes this name because the adjoint action of $\Pi$ on an arbitrary string is given by $\Pi \gamma_X \Pi =(-1)^{\abs{X}}\gamma_X$. Given $j\in [N]$, denote $V_j=\ci \Pi\gamma_j$.  One can easily verify that each $V_j$ is a self-adjoint unitary.  For any $A\in\cA$, consider the infinite-temperature fermionic \emph{Ornstein-Uhlenbeck} generator (see e.g., \cite[Section 6.1]{carlen2017gradient})
\be
\label{eq:majorana-ornstein-uhlenbeck}
    \cD(A)=\frac{1}{2}\sum_{j=1}^N(V_jAV_j-A).
\ee
For any monomial $\gamma_X$, one has $V_j \gamma_X V_j=(-1)^{\mathbf1\{j\in X\}}\gamma_X$, hence $\cD(\gamma_X)=-|X|\gamma_X=-\cN\gamma_X$. Thus $\cD$ acts diagonally on the Majorana basis with an eigenvalue corresponding to the degree.
\begin{remark}
\label{remark:pseudo-noise-fermi}
Note that if we define the noise operator
\[
\widetilde{\cD}(A) = \frac{1}{2} \sum_{j=1}^N (\gamma_j A \gamma_j - A)
\]
then one can verify that $\widetilde{\cD}(\gamma_X) = -|X| \gamma_X$ when $|X|$ is even, and otherwise $\widetilde{\cD}(\gamma_X) = -(N-|X|) \gamma_X$. The parity operator in the definition of $V_j = \ci \Pi \gamma_j$ ensures that on odd-degree monomials $\cD$ matches with $-\cN$.
\end{remark}

\subsection{Hamiltonians}
\label{sec:hamiltonians}

We consider Hermitian Hamiltonians $H$ of the form
\begin{equation}
  H=H_0+uV
  \label{eq:int-non-int-decomp}
\end{equation}
where $H_0$ is \emph{non-interacting} and $V$ is an arbitrary interacting term. In the Pauli case, a non-interacting Hamiltonian is a linear combination of single-site Pauli operators, while in the Majorana case a non-interacting Hamiltonian is an arbitrary linear combination of quadratic monomials.

It turns out that the main algebraic property that underlies these two cases is that if $P \in \cA_r$ is of degree $r$, then $[H_0,P]$ is also of degree $r$. In other words, $[H_0,\cA_r] \subset \cA_r$ for every $r$. On the other hand, the interacting part $V$ will in general increase the degree via commutation, i.e., $[V,P]$ can be of degree $\deg(P)+\delta$ for some $\delta > 0$. For example, if $V$ is a quadratic Pauli Hamiltonian, then one can check that $\deg([V,P]) \leq \deg(P) + 1$, while if $V$ is a quartic Majorana Hamiltonian, then $\deg([V,P]) \leq \deg(P) + 2$.

The following formalizes our assumptions on the Hamiltonian $H=H_0 + uV$ in the abstract algebraic setting:
\begin{assumption}[Hamiltonians]
\label{assumption:hamiltonian}
We assume $H$ can be written in the form $H=H_0 + uV$ with $0 \leq u \leq 1$, where $H_0$ satisfies
\begin{equation}
\label{eq:H0degreeinvariant}
[H_0,\cA_r]\subseteq\cA_r\qquad\forall r\in\{0,\dots,N\}.
\end{equation}
Furthermore, the interaction $V$ is assumed to satisfy
\begin{equation}
[V,\cA_r] \subseteq \cA_{\leq r+\delta} \quad \forall r \in \{0,\ldots,N\}
\end{equation}
for some constant $\delta$. Additionally, given the basis expansion $H = \sum_{P \in \cB} h_P P$, we assume that
\begin{itemize}
  \item Bounded coefficients: $|h_P|\le1$ for all $P$.
  \item $\Delta$-sparsity: for every $Q\in\BB$,
  \begin{equation}
    \bigl|\{P\in\BB:\ h_P\ne0,\ [P,Q]\ne0\}\bigr|\le\Delta\deg(Q).
    \label{eq:Delta-sparsity-assumption} 
  \end{equation}
\end{itemize}
\end{assumption}

The $\Delta$-sparsity assumption in \eqref{eq:Delta-sparsity-assumption} is related to the degree of the interaction (hyper)graph. For example if $H$ is a  Majorana Hamiltonian $H = \sum_{S \subset [N]} h_S \gamma_S$, then a sufficient condition for \eqref{eq:Delta-sparsity-assumption} is that\footnote{Moreover a Hamiltonian in the Majorana basis is always assumed to be of even degree}
\begin{equation}
\max_{i \in \{1,\ldots,N\}} |\{ S : i \in S \text{ and } h_S \neq 0 \}| \leq \Delta,
\end{equation}
see e.g., \cite[Proposition 2.3]{facelli2026fast}.

\subsection{Norms}

To quantify the accuracy of an estimator $\widehat{A}$ with respect to some target observable $A$, one can use several norms on $\cA$. We define here the norms that will be used for the different applications, and describe their operational interpretation.

\paragraph{Frobenius norm} In Section \ref{sec:weak-interaction}, we work with the normalized Frobenius norm defined, for any $A \in \cA$, as
\[
\norm{A}_\F = \sqrt{\tr(A^* A)}.
\]
Recall that from Section \ref{sec:norms-shared-setup} for any monomial $P \in \cB$ we have $\norm{P}_{\F} = 1$ and $\tr(P) = 0$ for all $P \notin \cP_0$. If the observable $A$ has an expansion
\[
A = \sum_{P\in \cB} a_P P
\]
then the Frobenius norm is given by
\be
\label{eq:frobA}
\norm{A}_\F = \left(\sum_{P\in\cB} |a_P|^2\right)^{1/2}.
\ee
While there are multiple interpretations of the Frobenius norm distance between two observables $A$ and $\widehat{A}$, we mention here one that fits particularly well within our context, also employed in \cite{facelli2026fast,schuster2025polynomial}. Considering some ensemble of states $\cE :=\{\rho_k\}_k$, which is assumed to be of \emph{low-average}, i.e. for some constant $c > 0$ we have
\[
    \norm{ \dfrac{1}{\abs{\cE}}\sum_{k} \rho_k }_\infty = \dfrac{c}{\sqrt{\abs{\cB}}} \,,
\]
where $\|\cdot\|_{\infty}$ is the operator norm.
This condition is satisfied for example if $\cE$ is an ensemble of pure states forming an orthonormal basis, in which case $c=1$.
A standard calculation shows that (see \cite[Section 2.1]{facelli2026fast} for details), for any $X\in\cA$ we have
\[
    \Ebb_{\rho\sim\cE} \bigg[ \abs{\Tr(\rho X)}^2 \bigg] \leq c \norm{X}^2_{\F}
\]
where $\Tr$ here denotes the standard trace. Obtaining a bound on the Frobenius norm would then imply an upper bound on the root mean squared error of the expectation value of $X$ over the ensemble of states $\cE$.

\paragraph{Fourier-$1$ norm} In Section~\ref{sec:pseudo-lindbladian}, we work with a normalized version of the Fourier-$1$ norm. 
For $X\in\mathcal A$ let 
\[
X=\sum_{P\in\BB}x_P P,\quad\qquad x_P:=\tr(PX)
\]
for $\tr$ the normalized trace. Then define the norm as:
\[
\|X\|_{\Fone}:=\sum_{P\in\BB}|x_P|.
\]
Recalling that $\cN$ is the degree operator, note that for any $X \in \cA$,
\[
\|e^{b\cN}(X)\|_{\Fone}  = \sum_{P\in\BB}e^{b\deg(P)}| x_P|.
\]
Exponentially-weighted norms will play a crucial role in the analyses of Sections \ref{sec:weak-interaction} and \ref{sec:pseudo-lindbladian}.
We finish this section with the following two simple lemmas which will be useful for the analyses. 
\begin{lemma}[Submultiplicativity of the weighted Fourier-$1$ norm]
\label{wth:lem:weighted-product}
For any operators $X,Y$ and $b\geq0$,
\begin{equation}
 \|e^{b\cN}(XY)\|_{\Fone}
 \leq
 \|e^{b\cN}X\|_{\Fone}
 \|e^{b\cN}Y\|_{\Fone}.
 \label{wth:eq:weighted-product}
\end{equation}
\end{lemma}

\begin{proof}
Write
\[
 X=\sum_{P\in\BB} x_P P,
 \qquad
 Y=\sum_{R\in\BB} y_R R.
\]
For both Pauli and Majorana monomials, $PR$ is another basis
monomial up to a phase and
\[
 \deg(PR)\leq\deg(P)+\deg(R).
\]
Hence, by the triangle inequality,
\[
 \begin{aligned}
 \|e^{b\cN}(XY)\|_{\Fone}
 &\leq
 \sum_{P,R}
 |x_P|\,|y_R|\,
 e^{b\deg(PR)}\\
 &\leq
 \sum_{P,R}
 |x_P|\,|y_R|\,
 e^{b\deg(P)}e^{b\deg(R)}\\
 &=
 \|e^{b\cN}X\|_{\Fone}
 \|e^{b\cN}Y\|_{\Fone}.
 \end{aligned}
\]
\end{proof}

\begin{lemma}
\label{lemma:inequalities}
Assume $H$ is a $\Delta$-sparse Hamiltonian (in the sense of \eqref{eq:Delta-sparsity-assumption}) satisfying
$|h_P|\leq1$ for every $P\in\BB$.
Then, for every $X\in\cA$,
\begin{equation}
 \|[H,X]\|_{\Fone}
 \leq2\Delta\deg(X)\|X\|_{\Fone}.
 \label{eq:bound[H,X]Fone}
\end{equation}
\end{lemma}

\begin{proof}
For any string $Q\in\BB$,
\begin{equation}
 [H,Q]
 =
 \sum_{\substack{P\in\BB\\
       h_P \neq0,\ [P,Q]\neq0}}
 h_P [P,Q].
\end{equation}
By $\Delta$-sparsity, there are at most
$\Delta\deg(Q)$ nonzero summands. Since
$\|[P,Q]\|_{\Fone}\leq2$ (see point (iii) in Section \ref{sec:norms-shared-setup}) and
$|\widehat H(P)|\leq1$, we obtain
\[
 \|[H,Q]\|_{\Fone}\leq2\Delta\deg(Q).
\]
Expanding $X$ in the basis and applying the triangle
inequality gives
\[
 \begin{aligned}
 \|[H,X]\|_{\Fone}
 \leq
 \sum_{Q\in\BB}
 |x_Q|\,\|[H,Q]\|_{\Fone}
 \leq
 2\Delta\sum_{Q\in\BB}
 \deg(Q)|x_Q|
 \leq
 2\Delta\deg(X)\|X\|_{\Fone}.
 \end{aligned}
\]
This proves Eq.~\eqref{eq:bound[H,X]Fone}.
\end{proof}

\subsection{Propagation algorithm and abstract analysis}
\label{sec:abstract}

Given $\ell\in \Nbb$, let $\Trunc_\ell:\cA\to\cA_{\le\ell}$ be the degree-$\ell$ truncation and let $(\cP_t)_{t\ge0}$ be the exact Heisenberg semigroup whose evolution we wish to approximate.  For a short step $h$ we choose an efficiently computable approximation $\widetilde{\cP}_h$.

For $t=mh$, a \emph{propagation} algorithm is simply
\be
 \label{eq:algorithm}
 A_0=A,
 \qquad
 A_{j+1}=\Trunc_\ell\widetilde{\cP}_h(A_j),
 \qquad j=0,\ldots,m-1.
\ee
The output $A_m$ is our approximation to $\cP_t(A)$.
The purpose of truncation is to keep the algorithm running efficiently. Indeed, an operator of degree at most $\ell$ has at most $N^{O(\ell)}$ coefficients. Thus the algorithm is useful when $\ell$ remains small.

The following theorem gives a simple analysis of the error accumulated by a propagation algorithm for an abstract norm $\|\cdot\|$.

\begin{theorem}[Propagation theorem]
\label{thm:abstract}
Let $\|\cdot\|$ be a norm on $\cA$, and consider the propagation algorithm \eqref{eq:algorithm}.
Assume that $\Trunc_\ell$ and $\cP_s$ are contractive in $\|\cdot\|$, i.e., for any $X \in \cA$
\begin{equation}
\|\Trunc_{\ell} X\| \leq \|X\| \text{ and }
\|\cP_s X\| \leq \|X\| \quad \forall s \geq 0.
\end{equation}
Assume also that for every $X\in\cA_{\le\ell}$,
\begin{align}
 \|\widetilde{\cP}_hX-\cP_hX\|
 &\le \delta_\ell(h)\|X\|
 \label{eq:local-error}
\end{align}
for some function $\delta_{\ell}(h)$. Let $A \in \cA_{\leq \ell}$.
Then the output of the iterations \eqref{eq:algorithm} satisfies
\be
 \label{eq:total-error}
 \|A_m-\cP_{mh} A\|
 \le
 m\left[\delta_\ell(h)(1+\delta_{\ell}(h))^m+\eta_\ell(t;A)\right]\|A\|
\ee
where
\begin{equation}
 \eta_\ell(t;A)
 :=\sup_{0\le s\le t}
 \frac{\|(I-\Trunc_\ell)\cP_s(A)\|}{\|A\|}
 \label{eq:eta}
\end{equation}
is the exact truncation tail up to time $t$ of $\cP_s(A)$.
\end{theorem}

\begin{proof}
Let $B_j=\cP_{jh}A$.  Contractivity of $\Trunc_\ell$ and of $\cP_h$ gives
\begin{equation}
\label{eq:generic-recurrence}
\begin{aligned}
 \|A_j-B_j\| &= \|\Trunc_{\ell} \widetilde{\cP}_h A_{j-1} - \cP_h B_{j-1}\|\\
 &\le \|\widetilde{\cP}_hA_{j-1}-\cP_hB_{j-1}\|
      +\|(I-\Trunc_\ell)B_j\|\\
 &\le \delta_\ell(h)\|A_{j-1}\|
      +\|A_{j-1}-B_{j-1}\|
      +\eta_\ell(t;A)\|A\|.
\end{aligned}
\end{equation}
Note that for any $X \in \cA_{\leq \ell}$
\[
\|\Trunc_{\ell} \widetilde{\cP}_h X\| \leq \|\widetilde{\cP}_h X\| \leq \|\cP_h X\| + \|(\cP_h - \widetilde{\cP}_h) X\| \leq (1+\delta_{\ell}(h)) \|X\|.
\]
As such, $\|A_j\|\le(1+\delta_\ell(h))^j\|A\|$.  Iterating the recurrence \eqref{eq:generic-recurrence} from $A_0=B_0=A$ yields
\[
 \|A_m-B_m\|
 \le \delta_\ell(h)\sum_{j=0}^{m-1}(1+\delta_\ell(h))^j\|A\|
      +m\eta_\ell(t;A)\|A\|,
\]
and the geometric sum is at most $m(1+\delta_\ell(h))^m$.
\end{proof}

The theorem leads to a particularly simple sufficient condition for efficiency of the algorithm.

\begin{corollary}[Exponential tail criterion]
\label{cor:criterion}
Suppose $A$ has degree at most $d$.  Assume, for some polynomial $q$ and positive constants $a,C$ independent of $N,t,\ell$,
\begin{equation}
 \delta_\ell(h)\le q(\ell)h^2,
 \qquad
 \eta_\ell(t;A)\le C e^{-a(\ell-d)}
 \label{eq:criterion}
\end{equation}
for all sufficiently small $h$, where $\eta_{\ell}$ is as defined in \eqref{eq:eta}.  Then for every $t$ and $0<\eps<1$ one can choose
\begin{equation}
 \ell=d+O\!\left(a^{-1}\log\frac{1+t}{\eps}\right),
 \qquad
 m=\poly\!\left(t,\frac1\eps,q(\ell)\right),
 \label{eq:criterion-choice}
\end{equation}
so that the output $A_m$ of the algorithm \eqref{eq:algorithm} with $h=t/m$ satisfies
\[
\|A_m - \cP_t A\| \leq \eps \|A\|.
\]
If a degree-$\ell$ observable can be stored and updated in $N^{O(\ell)}$ operations, then the total running time is
\begin{equation}
 N^{O\left(d+a^{-1}\log((1+t)/\eps)\right)}\poly(t,1/\eps).
 \label{eq:criterion-runtime}
\end{equation}
\end{corollary}

\begin{proof}
Take $h=t/m$.  Choose $m$ large enough so that
\begin{align}\label{eq:h-m-criteria}
 m q(\ell)h^2=\frac{q(\ell)t^2}{m}\le \min\{1,\eps/(4e)\}.
\end{align}
Then $(1+q(\ell)h^2)^m\le e$ and the accumulated short-time error in \eqref{eq:total-error} is at most $\eps\|A\|/4$.  The truncation contribution is at most
$Cm e^{-a(\ell-d)}\|A\|$.  Requiring this to be at most $\eps\|A\|/2$ gives
\[
 \ell-d\ge a^{-1}\log\frac{2Cm}{\eps}.
\]
Because $m$ is polynomial in $t,1/\eps$ and in the polynomial $q(\ell)$, the extra $\log q(\ell)=O(\log(1+\ell))$ is absorbed by enlarging the constant in \eqref{eq:criterion-choice}.  This concludes the proof.
\end{proof}

\section{Hamiltonian evolution with dissipation}
\label{sec:weak-interaction}

In this section, we study propagation algorithms to simulate the evolution generated by a Lindbladian of the form
\begin{equation}
 \label{eq:dissipative-generator}
 \cL(A)=\ci[H_0+uV,A]-\lambda\cN A,
 \qquad u,\lambda\geq0,
\end{equation}
where $H=H_0+uV$ satisfies Assumption \ref{assumption:hamiltonian} and $-\lambda \cN$ is the generator of a depolarizing noise (see Section \ref{sec:degree-operator}). When $\lambda=0$ this corresponds to exact unitary evolution, which was the subject of \cite{facelli2026fast}. We call $\cP_t = e^{t\cL}$ the semigroup generated by $\cL$.

\subsection{Contraction}
\label{sec:dissipative-dynamics-truncation}

As highlighted in Section \ref{sec:abstract}, the main ingredient to obtain guarantees on the propagation algorithm is to prove contraction of $\cP_t$ for an exponentially-weighted norm. In this section we work with the Frobenius norm, and the exponentially-weighted norm of an observable $X = \sum_{P \in \cB} x_P P \in \cA$ can be written as
\[
\|e^{a\cN} X\|_{\F} = \left(\sum_{P \in \cB} (e^{a\deg(P)} |x_P|)^2\right)^{1/2}.
\]
Proving contraction of $\cP_t$ for such an exponentially weighted norm amounts to showing that
\begin{equation}
\label{eq:contractioneaN}
\frac{\dd}{\dd t} \|e^{a \cN} \cP_t X\|^2_{\F} \leq 0.
\end{equation}
The goal of the next proposition is to identify conditions on $a$ such that \eqref{eq:contractioneaN} holds. To allow finer control on this contractivity, we in fact deal with the case where the weight parameter $a > 0$ can be time-varying.

\begin{theorem}[Contraction of the weighted Frobenius norm]
\label{thm:dissipative-weighted-contraction}
    Let $\cL$ be the Lindbladian \eqref{eq:dissipative-generator}. Assume $a: \RR^+\to \Rbb^+$ is a positive function satisfying for all $t \geq 0$
    \be
    \label{eq:ode-for-a(t)}
        a'(t) \leq \lambda - 2\Delta\abs{u}\sinh(a(t)\delta).
    \ee
    Then, for any $X\in \cA$, we have
    \be
    \label{eq:weighted-contraction}
        \frac{\dd}{\dd t} \|e^{a(t)\cN} \cP_t X\|_{\F}^2\le0.
    \ee
\end{theorem}

Before proving the theorem, note that $a(t) \equiv 0$ satisfies \eqref{eq:ode-for-a(t)}, which shows that $\cP_t$ is a contraction for the (unweighted) Frobenius norm. Setting $a(t) = a$ constant, we see that $\cP_t$ is a contraction for the $a$-weighted Frobenius norm when $2 \Delta |u| \sinh(a\delta) \leq \lambda$.

\begin{proof}
Let $W_t=e^{a(t)\cN}$, $X_t=e^{t\cL}X$, and $Y_t=W_t X_t$. To prove \eqref{eq:weighted-contraction}, we show that the derivative of $\|Y_t\|^2_F$ is nonpositive.

First of all, notice that
\[
\begin{aligned}
\frac12 \frac{\dd}{\dd t} \norm{Y_t}^2_{\F} 
&=\frac12 \left(\HSproduct{\frac{\dd}{\dd t} Y_t}{Y_t} + \HSproduct{Y_t}{\frac{\dd}{\dd t} Y_t} \right),
\end{aligned}
\]
and
\[
\begin{aligned}
(\dd/\dd t) Y_t = a'(t) \cN Y_t + W_t \cL X_t = a'(t) \cN Y_t + W_t \cL W_t^{-1} Y_t.
\end{aligned}
\]
Hence we can write
\be
\label{eq:ddtnormYt^2}
\begin{aligned}
\frac12 \frac{\dd}{\dd t} \norm{Y_t}^2_F 
&= a'(t)\HSproduct{Y_t}{\cN Y_t} + \frac12 \bigg(\HSproduct{Y_t}{W_t \cL W_t^{-1} Y_t} + \HSproduct{W_t \cL W_t^{-1} Y_t}{ Y_t}\bigg)  \\
&= a'(t)\HSproduct{Y_t}{\cN Y_t} + \HSproduct{ Y_t}{ \frac{\cL_t +\cL_t^\dagger}{2} Y_t}
\end{aligned}
\ee
where $\cL_t = W_t\cL W_t^{-1}$. Recall that
\[
\cL = \ci [H_0 + uV, \cdot] - \lambda \cN = \cL_0 + u \cG - \lambda \cN
\]
where we let $\cL_0 = \ci [H_0, \cdot]$, and $\cG = \ci [V,\cdot]$. We have the following:
\begin{itemize}
    \item Since $\cN$ commutes with $W_t = e^{a(t) \cN}$, we have $\cL_t = W_t (\cL_0 + u \cG) W_t^{-1} - \lambda \cN$.
    \item Furthermore, note that $\cL_0$ is degree-preserving, i.e., if $P$ is a degree-$r$ monomial, then $\cL_0(P)$ is homogeneous of degree $r$. This means that $\cL_0$ commutes with $\cN$ and thus with $W_t$. Hence we can write $\cL_t = u W_t \cG W_t^{-1} + \cL_0 - \lambda \cN$. Note additionally that $\cL_0$ is skew-Hermitian, i.e., that $\cL_0 + \cL_0^{\dagger} = 0$, and so the contribution of $\cL_0$ vanishes in  $\cL_t + \cL_t^{\dagger}$.
\end{itemize}

From the above, we have
\be
\label{eq:identityddtYt2}
\frac12 \frac{\dd}{\dd t} \norm{Y_t}^2_F = (a'(t) - \lambda) \HSproduct{Y_t}{\cN Y_t} + u\HSproduct{Y_t}{\frac{\cG_t + \cG_t^{\dagger}}{2} Y_t}
\ee
where $\cG_t = W_t \cG W_t^{-1}$. We now use the following lemma.

\begin{lemma}
$|\HSproduct{Y_t}{\frac{\cG_t + \cG_t^{\dagger}}{2} Y_t}| \leq 2 \Delta \sinh(a(t) \delta) \HSproduct{Y_t}{\cN Y_t}$.
\end{lemma}
\begin{proof}
We use Lemma \ref{lemma:master-commutator-bound}. Note that $\cG_t = W_t \cG W_t^{-1} = \sum_{r,s \geq 0} e^{(s-r)a(t)} \P_s \cG \P_r$. Using the fact that $\cG$ is skew-Hermitian, we get
\[
\cG_t + \cG_t^{\dagger} = \sum_{r,s \geq 0} (e^{(s-r)a(t)} - e^{(r-s)a(t)}) \P_s \cG \P_r = 2 \sum_{r,s \geq 0} \phi(s-r) \P_s \cG \P_r = 2 \cG_{\phi}
\]
with $\phi(r) = \sinh(ra(t))$, with the notation of Lemma \ref{lemma:master-commutator-bound}.
Inequality \eqref{eq:master} immediately implies the result.
\end{proof}

Therefore from \eqref{eq:identityddtYt2} we get
\be
 \frac{1}{2} \frac{\dd}{\dd t}\|Y_t\|_{\F}^2
 \leq \bigg[a'(t)-\lambda+2\Delta|u|\sinh(\delta a(t))\bigg]
 \HSproduct{Y_t}{\cN Y_t}  \leq 0
\ee
where we used the assumption \eqref{eq:ode-for-a(t)} on $a(t)$ and the fact that $\HSproduct{Y_t}{\cN Y_t} \geq 0$. This proves \eqref{eq:weighted-contraction}.
\end{proof}

The following theorem exploits the contraction property proved above to obtain an upper bound on the tail of the truncation of $\cP_t$, i.e., to obtain bounds on $\|\cP_t X - \Trunc_{\ell} \cP_t X\|_{\F}$.

\begin{theorem}[Low-degree truncation of weakly-interacting dynamics]
\label{thm:weak-dissipative-trunc}
Let $\cL$ be the Lindbladian  defined in \eqref{eq:dissipative-generator} and consider a degree-$d$ observable $X\in\cA$. Then we have that, for any $\ell\geq d$ and $t \leq \tmax(u,\lambda)$
\be
\label{eq:truncation-bound}
    \norm{(I-\Trunc_\ell)\cP_t X}_{\F} \leq \kappa e^{-(\ell-d)/\delta}\norm{X}_F, \qquad \kappa = e^{d/\delta}
\ee
where $\tmax(u,\lambda)$ is given by
\be
\label{eq:tmax-def}
    \tmax(u,\lambda) =\begin{cases}
        \dfrac{1}{2\Delta\abs{u}\delta \cosh(1)} \log\left(1+ \dfrac{2\Delta\abs{u}c_0}{\lambda_c(u)-\lambda}\right) \quad &\lambda< \lambda_c(u)\\
        \infty \quad &\lambda\geq\lambda_c(u)
    \end{cases}
\ee
where $\lambda_c(u):=2\Delta\abs{u} \sinh(1)$ and $c_0=2(e-1)/(e^2+1)$.
\end{theorem}
\begin{remark*}
    Let us study the value $\tmax(u,\lambda)$ at some special points:
    \begin{itemize}
        \item  For $\lambda=0$, this corresponds to the closed-system case, and similarly to \cite{facelli2026fast} we expect a finite maximal time which increases as $u$ gets smaller. Indeed, we find
        \[
            \tmax(u,\lambda=0) = \dfrac{\log(1+c_0/\sinh(1))}{2\Delta \abs{u} \delta \cosh(1)}
        \]
        which indeed has a $O(1/\abs{u})$ scaling, an exponential improvement compared to the results in \cite{facelli2026fast} and matching the scalings of \cite{zhao2026efficient}.
        \item When $u=0$ note that $\lambda_c=0$ and so $\tmax(0,\lambda)=+\infty$ for all $\lambda \geq 0$.
        \item For a fixed value of $u > 0$, we note that $t_{\max}(u,\lambda)$ grows to $+\infty$ as $\lambda \uparrow \lambda_c(u)$ like $\log\left(\frac{1}{\lambda_c(u)-\lambda}\right)$.
    \end{itemize}
\end{remark*}
\begin{proof}
When $\lambda \geq \lambda_c(u)$, we use Theorem \ref{thm:dissipative-weighted-contraction}, with $a(t) \equiv 1/\delta$ which satisfies \eqref{eq:ode-for-a(t)} for all $t \geq 0$ since $a'(t) = 0 \leq \lambda - 2\Delta|u|\sinh(1) = \lambda - \lambda_c$. This shows \eqref{eq:truncation-bound} with $\kappa=1$ in the regime $\lambda \geq \lambda_c(u)$.

We now assume $\lambda < \lambda_c(u)$. 
We use Theorem \ref{thm:dissipative-weighted-contraction}, where we take $a(t)$ to be a solution of the differential equation obtained from \eqref{eq:ode-for-a(t)}:
\be
\label{eq:odea}
a'(t) = \lambda - k\sinh(\delta a(t)),\qquad k=2\Delta |u|.
\ee
If the solution satisfies $a(t) \geq a_*$ for all $t \in [0,T]$, then we immediately get that from any $X$ of degree $d$,
\[
\norm{(I-\Trunc_\ell)e^{t\cL}X}_{\F} \leq e^{d(a(0) - a_*)} e^{-a_* (\ell-d)} \norm{X}_{\F} \qquad \forall t \in [0,T],
\]
since
\[
\norm{(I-\Trunc_\ell)e^{t\cL}X}_{\F} \leq e^{-a(t)\ell}\norm{e^{a(t)\cN}e^{t\cL}X}_{\F} \leq e^{-a(t)\ell}\norm{e^{a_0\cN}X}_{\F} \leq e^{-a(t)\ell+a_0 d}\norm{X}_{\F}.
\]
We will choose $a(0)=2/\delta$, $a_*=1/\delta$. It can be shown that the solution of the ODE \eqref{eq:odea} is given by
\be
\label{eq:solution-a(t)}
    a(t) = \frac{1}{\delta}\log\left(\dfrac{e^\alpha + C e^{-\alpha} e^{-\Omega t}}{1-Ce^{-\Omega t}}\right) \,,
\ee
with constants
\be
\label{eq:constants-in-a(t)}
\alpha=\arsinh(\lambda/k)\,, \quad C = \dfrac{e^{\delta a(0)}-e^\alpha}{e^{\delta a(0)}+e^{-\alpha}} \,, \quad \Omega=\delta\sqrt{k^2+\lambda^2}
\ee
Note that since $\lambda < \lambda_c = k\sinh(1)$, we have $\alpha < 1$, and since $a(0) = 2/\delta$ we have $C > 0$. (One way to arrive at \eqref{eq:solution-a(t)} is via the change of variables $z(t) = e^{\delta a(t)}$ which turns the ODE \eqref{eq:odea} into a Riccati equation. We omit the details here.) Note that $a'(t) \leq 0$ for all $t \geq 0$; indeed
\[
\frac{1}{k} a'(t) = \lambda/k - \sinh(\delta a(t)) = \sinh(\alpha) - \sinh(\delta a(t)) < 0
\]
where we used the fact that $\delta a(t) > \alpha$ which easily follows from \eqref{eq:solution-a(t)} since $C > 0$. Furthermore, note that $\lim_{t\to\infty} a(t) = \alpha/\delta < 1/\delta := a_*$.
Thus there exists a unique $t_*$ such that $a(t_*)=a_*$. From \eqref{eq:solution-a(t)} one gets:
    \[
         t_* = \frac{1}{\Omega}\log\left(\dfrac{(e^2-e^\alpha)(e+e^{-\alpha})}{(e^2+e^{-\alpha})(e-e^\alpha)}\right).
    \]
We now simplify this expression to make the dependences on $u$ and $\lambda_c - \lambda$ more explicit. Denoting $x:=e^\alpha < e$, and isolating a factor of $1$ in the logarithm we equivalently have
    \[
        t_* = \frac{1}{\Omega}\log\left(1 + \dfrac{(x^2+1)}{e^2x+1}\dfrac{e(e-1)}{e-x} \right)
    \]
    Let $f(x)=(x^2+1)/(e^2x+1)$. Since $f'(x)\geq 0 $ for all $x\geq 1$, we use $f(x)\geq f(1)$ and obtain
    \[
        t_* \geq \frac{1}{\Omega}\log\left(1 + \dfrac{2e(e-1)}{e^2+1}\dfrac{1}{e-x} \right)
    \]
    Note also that
    \[
    \begin{aligned}
        e-e^\alpha &= e(1-e^{-(1-\alpha)}) \\
        &\leq e(1-\alpha) \\
        &= e(\arsinh(\lambda_c/k)-\arsinh(\lambda/k))
    \end{aligned}
    \]
    where the inequality is a consequence of $1-e^{-z}\leq z$ for $z\in\Rbb$ and the last equality is obtained by definition of $\alpha$ and $\lambda_c$. Finally, the following inequality
    \[
    \arsinh(\lambda_c/k)-\arsinh(\lambda/k) =\int_\lambda^{\lambda_c} \dd s\dfrac{1}{\sqrt{k^2+s^2}} \leq \dfrac{\lambda_c-\lambda}{\sqrt{k^2+\lambda^2}}
    \]
    implies that
    \[
        t_* \geq \dfrac{1}{\Omega} \log\left(1+\dfrac{2(e-1)}{e^2+1}\dfrac{\sqrt{\lambda^2+k^2}}{\lambda_c-\lambda}\right) 
    \]
    Noting that $\sqrt{\lambda^2+k^2}\leq \sqrt{\lambda_c^2+k^2}=k\cosh1$ as well as $\sqrt{\lambda^2+k^2}\geq k$, we have
    \[
        t_* \geq \dfrac{1}{\delta k\cosh(1)} \log\left(1+\dfrac{2(e-1)}{e^2+1}\dfrac{k}{\lambda_c-\lambda}\right) 
    \]
    Recalling that $k=2\Delta \abs{u}$, the right-hand side corresponds to \eqref{eq:tmax-def} in the regime $\lambda<\lambda_c$. This completes the proof.
\end{proof}

\subsection{Time simulation}

Combining the theorem above with Theorem \ref{thm:abstract} allows us to get guarantees on propagation algorithms to simulate the Lindbladian \eqref{eq:dissipative-generator}. Given a time step $h > 0$, we consider the Euler first order approximation of $\cP_h$, defined by:
\be
 \widetilde{\cP}_h=I+h\cL.
\ee
To simulate the time evolution up to a horizon $t > 0$ for an initial observable $A \in \cA$, we consider the algorithm
\be
\label{eq:propagation-algo-weak}
\begin{cases}
A_0 = A\\
A_{j+1} = \Trunc_{\ell} \widetilde{\cP}_{h} A_j \qquad (j=0,\ldots,m-1),
\end{cases}
\ee
for some $m \in \NN$, where $h = t/m$.
The next corollary shows that this algorithm with 
\be
\label{eq:weak-interaction-complexity}
    \ell=d+O\!\left(\log\frac{1+t}{\eps}\right),
    \qquad
    m=\poly(t,1/\eps)
\ee
allows us to efficiently approximate $\cP_t A$ provided $t < \tmax(u,\lambda)$ where $\tmax(u,\lambda)$ is defined in \eqref{eq:tmax-def}.

\begin{corollary}
\label{cor:weak-dissipative-propagation}
Let $\cL$ be the Lindbladian as defined in \eqref{eq:dissipative-generator}, and fix an arbitrary initial observable $A \in \cA$, and a time horizon $0 < t < \tmax(u,\lambda)$, where $\tmax(u,\lambda)$ is defined in \eqref{eq:tmax-def}.

Then for any $\eps >0$, there is an algorithm that runs in time $N^{O(\log ((1+t) /\eps))} \poly(t,1/\eps)$ and returns an approximation $\tilde{A}$ of $\cP_t A$ such that $\|\cP_t A - \tilde{A}\|_{\F} \leq \eps \|A\|_{\F}$. The approximant $\tilde{A}$ is represented as a polynomial of degree at most $\ell = d + O(\log ((1+t)/\eps))$.
\end{corollary}
\begin{proof}
We apply Corollary \ref{cor:criterion} with the Frobenius norm.
Using Lemma \ref{lemma:euler-truncation-error}, we have that for $X\in\cA_{\leq \ell}$
\begin{align}
 \|\cP_h X-\widetilde{\cP}_h X\|_{\F}
 &\le c\cdot (\ell+\delta)^2 h^2\|X\|_{\F},
 \label{eq:weak-interaction-local}
\end{align}
where $c$ is a constant depending only on $\Delta$ and $\lambda$.
From Theorem \ref{thm:weak-dissipative-trunc} we get that, with the Frobenius norm $\eta_{\ell}(t;A) \leq e^{d/\delta} e^{-(\ell-d)/\delta}$ for all $0 \leq t \leq \tmax(u,\lambda)$ (where $\eta$ is defined in \eqref{eq:eta}).

\end{proof}

\subsection{Spin and Majorana applications}
For concreteness, we consider two specific instances whose out-of-equilibrium dynamics can be efficiently simulated by the propagation algorithm.
\paragraph{Spin systems}
Consider the algebra $\cA$ of spins with system size $N$. The basis of the algebra is described in \eqref{eq:pauli-algebra}.

Consider a Hamiltonian $H=H_0+uV$ where $H_0$ acts on single sites and $V$ is a $\Delta$-sparse interaction between pairs of sites, namely
\[
    H_0 = \sum_{\substack{P\in\BB\,\\ \deg(P)=1}} h_P P\,,\qquad V=\sum_{\substack{P\in\BB\\ \deg(P)=2}} h_P P
\]
For arbitrary $\lambda >0$, we then consider the following Lindbladian
\be
\label{eq:spin-depolarising-lindbladian}
    \cL(X) = \ci[H_0+uV,X] + \lambda \cD(X)
\ee
where $\cD$ is the depolarising noise introduced in \eqref{eq:depolarizing-noise-definition}. We have the following result regarding efficient simulability of spin system dynamics
\begin{corollary}[Efficient simulation of weakly-interacting spin systems]
\label{cor:weak-dissipative-dynamics-spin}
    Consider a system of $N$ spins and let $\cL$ be given by \eqref{eq:spin-depolarising-lindbladian}. Let $A\in\cA$ be an arbitrary degree-$d$ observable. Fix a time horizon $0\leq t \leq \tmax(u,\lambda)$ where
    \[
        \tmax(u,\lambda) =
        \begin{cases}
        \dfrac{1}{2\Delta u \cosh(1)}\log\left(1+\dfrac12\dfrac{\Delta u }{\lambda_c(u)-\lambda}\right) \quad &\lambda < \lambda_c(u) \\
        \infty \quad &\lambda\geq \lambda_c(u)
        \end{cases}
    \]
    with $\lambda_c(u)=2\Delta u \sinh(1)$. Given $0<\eps<1$, the propagation algorithm in the Pauli basis produces a degree-$\ell$ approximant $\widetilde{A}$ s.t. $\norm{\tilde{A}-e^{t\cL}A}_\F \leq \eps\norm{A}_\F$ in complexity $N^{O(\log((1+t)/\eps))}$ with $\ell = d+O(\log((1+t)/\eps))$
\end{corollary}
\begin{proof}
    The spin algebra satisfies the required properties described at the start of Section \ref{sec:weak-interaction} and thus the result is a direct consequence of Corollary \ref{cor:weak-dissipative-propagation} with $\delta=1$
     and $c_0 \geq 1/4$ in \eqref{eq:tmax-def} of $\tmax(u,\lambda)$.
\end{proof}

\paragraph{Fermionic systems}
Consider now the case of a fermionic algebra $\cA$ on $N/2$ sites, where $N$ is assumed even. Let $\{\gamma_1,\dots,\gamma_N\}$ denote the corresponding self-adjoint Majorana operators and take the corresponding basis defined in \eqref{eq:majorana-basis}. 

Consider a $\Delta$-sparse Hamiltonian $H=H_0+uV$ where $H_0$ is the quadratic part and $V$ a quartic part:\footnote{Note that typical fermionic Hamiltonians (e.g. the Fermi-Hubbard model) can be cast in the form above due to the linear mapping $c_x=\frac{1}{2}(\gamma_{2x-1}+\ci \gamma_{2x})$ and $c_x^*=\frac{1}{2}(\gamma_{2x-1}-\ci \gamma_{2x})$ where $c_x,c_x^\dagger$ are the second quantization operators for any site $x\in \{1,\dots,N/2\}$.}
\[
    H_0 = \sum_{\substack{X\subset[N] \\ \abs{X}=2}} h_X \gamma_X \,,\qquad V = \sum_{\substack{X\subset[N] \\ \abs{X}=4}} h_X \gamma_X
\]
Similarly to the spin example, for arbitrary $\lambda>0$ consider the following Lindbladian
\be
\label{eq:fermion-depolarising-lindbladian}
    \cL(X) =\ci[H_0+uV,X] +\lambda \cD(X)
\ee
where $\cD$ is the fermionic Ornstein-Uhlenbeck generator introduced in \eqref{eq:majorana-ornstein-uhlenbeck}.
We can then show the following
\begin{corollary}[Efficient simulation of weakly-interacting fermionic systems]
\label{cor:weak-dissipative-dynamics-fermionic}
Consider a system of $N$ Majorana modes  and let $\cL$ be given by \eqref{eq:fermion-depolarising-lindbladian}. Let $A\in\cA$ be an arbitrary degree-$d$ observable. Fix a time horizon $0\leq t \leq \tmax(u,\lambda)$ where
    \[
        \tmax(u,\lambda) =
        \begin{cases}
        \dfrac{1}{4\Delta u \cosh(1)}\log\left(1+\dfrac{1}{2}\dfrac{\Delta u}{\lambda_c(u)-\lambda}\right) \quad &\lambda < \lambda_c(u) \\
        \infty \quad &\lambda\geq \lambda_c(u)
        \end{cases}
    \]
    with $\lambda_c(u)=2\Delta u \sinh(1)$. Then, for any $0<\eps<1$, the propagation algorithm in the Majorana basis produces a degree-$\ell$ approximant $\widetilde{A}$ s.t. $\norm{\tilde{A}-e^{t\cL}A}_\F \leq \eps\norm{A}_\F$ in complexity $N^{O(\log((1+t)/\eps))}$ with $\ell = d+O(\log((1+t)/\eps))$
\end{corollary}
\begin{proof}
        The Majorana algebra satisfies the required properties described at the start of Section \ref{sec:weak-interaction} and thus the result is a direct consequence of Corollary \ref{cor:weak-dissipative-propagation} with $\delta=2$
     and $c_0 \geq 1/4$ in the definition \eqref{eq:tmax-def} of $\tmax(u,\lambda)$.
\end{proof}

\section{Thermal observable estimation}
\label{sec:pseudo-lindbladian}

In this section we focus on the problem of estimating thermal expectation values of weakly interacting Hamiltonians. Given a Hamiltonian $H$, an inverse temperature $\beta\ge 0$ and an observable $Q$,
our goal is to estimate the Gibbs expectation
\[
  \operatorname{Tr}(\rho_{\beta,H}Q),\qquad
  \rho_{\beta,H}:=\frac{e^{-\beta H}}{\operatorname{Tr}(e^{-\beta H})}.
\]
For concreteness, we consider in this section fermionic Hamiltonians written in terms of $N$ Majorana operators.\footnote{The Pauli case can also be treated using the same techniques. We omit it here for clarity of exposition. In fact the Pauli case is simpler to handle since then the Hamiltonian $H_0$ is a linear combination of single-site Paulis, in which case one doesn't need an additional assumption like Assumption \ref{assumption:diagonal-hamiltonian-fermion}.} We assume $H=H_0+uV$ decomposes into non-interacting and interacting parts as in
Assumption \ref{assumption:hamiltonian}. We will impose further assumptions on $H_0$ that will be detailed later.

The main result of this section is a deterministic
algorithm, based on Majorana propagation, to approximate $\Tr(\rho_{\beta,H} Q)$ to within $\eps$, which runs in time $(N+1)^{O(s+\log(1/\eps))}\operatorname{poly}(s,1/\eps)$ when $Q$ is a degree-$s$ monomial, and
whenever $\beta$ and $u$ satisfy an explicit high-temperature/weak-interaction condition of the form $u e^{\beta h} \ll 1$.

The approach consists of constructing a (generally non-completely-positive) semigroup $\cP_t=e^{t\mathcal K^\beta(H)}$ generated by a \emph{pseudo-Lindbladian} $\cK^{\beta}(H)$
that is unital and whose dual fixes $\rho_{\beta,H}$. If $\cP_tQ$ converges to
$\Tr(\rho_{\beta,H}Q)\,I$, then the Gibbs expectation is exactly the
coefficient of the identity in $\cP_tQ$ for large enough $t$. We estimate this
coefficient by applying the propagation framework (Theorem~\ref{thm:abstract}) to short steps
$I+h\mathcal K^\beta_M(H)$, where $\mathcal K^\beta_M(H)$ is a suitable computationally efficient approximation of $\mathcal K^\beta(H)$ based on Dyson expansion. Both fast convergence to the Gibbs state and the ability to apply Theorem~\ref{thm:abstract} require showing that
$\cP_t$ is exponentially contractive in an exponentially weighted Fourier-1 norm, which is where
the parameter condition enters.
Throughout this section we work with the Fourier 1-norm which we recall is defined as
\[
\|X\|_{\Fone} = \sum_{P \in \cB} |x_P|
\]
for an observable $X = \sum_{P \in \cB} x_P P$ expressed in the Majorana basis $\cB$. For convenience, we also set, for $a > 0$
\begin{equation}
\label{eq:exponentially-weighted-1-norm}
\|X\|_a := \|e^{a \cN} X\|_{\Fone} = \sum_{P \in \cB} e^{a\deg(P)} |x_P|.
\end{equation}

\subsection{Overview}
\label{sec:overviewgibbs}

We start by introducing the pseudo-Lindbladian that we study, and give an overview of the analysis.

\paragraph{A family of pseudo-Lindbladians}
For an operator $B \in \cA$ define the imaginary-time evolution by
\begin{equation}
\tau^H_{s}(B) := e^{s H} B e^{-sH}.
\end{equation}
Given operators $A_1,\ldots,A_{\g},B_1,\ldots,B_{\g} \in \cA$, we define the super-operator  $\cK^{\beta}(H): \cA \to \cA$ by
\begin{equation}
\label{eq:def-Kbeta}
\cK^{\beta}(H)X = \alpha \sum_{i=1}^{\g} \bigl([A_i,X]\,\tau^H_{-\beta/2}(B_i) - \tau^H_{\beta/2}(B_i) \,[A_i,X]\bigr),
\end{equation}
where $\alpha \geq 0$ is some suitable constant, and we let $\cP_t$ be the associated semigroup
\begin{equation}
\label{eq:semigroup}
\cP_t = e^{t \cK^{\beta}(H)}.
\end{equation}
An easy calculation shows  that the adjoint of $\cK^{\beta}(H)$ for the Hilbert-Schmidt inner product is given by
\begin{equation}
\label{eq:Kbeta-dagger}
(\cK^{\beta}(H))^{\dagger} \sigma = \alpha\sum_{i=1}^{\g} \Bigl[A^*_i \, , \, \sigma \tau^H_{\beta/2}(B_i^*) - \tau^H_{-\beta/2}(B^*_i) \sigma\Bigr].
\end{equation}
Since $\rho_{\beta,H}\tau^H_{\beta/2}(B_i^*) = \tau^H_{-\beta/2}(B^*_i) \rho_{\beta,H}$ for any $i$, we see that $\rho_{\beta,H} \in \ker (\cK^{\beta}(H))^{\dagger}$, and hence it is a fixed point of the adjoint semigroup $(\cP_t)^{\dagger}$. Thus if $Q \in \cA$ we see that
\[
\frac{\dd}{\dd t} \Tr((\cP_t Q) \rho_{\beta,H} ) = 0.
\]
Hence if $\cP_t Q \to cI$ as $t \to \infty$ for some constant $c$ (note that $I$ is a fixed point of $\cP_t$ since it is in $\ker \cK^{\beta}(H)$), it must be that $c = \Tr(\rho_{\beta,H} Q)$.

\paragraph{Analysis} The main goal of this section is to analyze the propagation algorithm to simulate the evolution $\cP_t Q$ for some initial observable $Q$.
To analyze the propagation algorithm for the evolution $\cP_t Q$, we need to show, as per the abstract analysis of Section \ref{sec:abstract}, that $\cP_t$ is a contraction for a suitable exponentially-weighted norm of the form \eqref{eq:exponentially-weighted-1-norm}. This amounts to showing that
\begin{equation}
\label{eq:decay-exponentially-weighted-norm}
\frac{\dd}{\dd t} \|\cP_t Q\|_{a} \leq 0 \qquad (t \geq 0)
\end{equation}
for some $a>0$. It is possible to show that, in general, $\cP_t$ is \emph{not} a contraction for the unweighted 1-norm. However we will show that there is a range $a \in [a_0,a_0+v]$ where contraction will occur.
This will allow us to obtain an exponential bound on the tail $\|(I-\Trunc_{\ell}) \cP_t Q\|_{a_0}$ of the form $e^{-v(\ell-\deg Q)}$,
which is the crucial condition in Theorem \ref{thm:abstract} for the convergence of the propagation algorithm. Furthermore, condition \eqref{eq:decay-exponentially-weighted-norm} will also guarantee that the evolution $\cP_t Q$ converges exponentially fast to $\Tr(Q \rho_{\beta,H}) I$, which implies that we only need to simulate $\cP_t Q$ up to time $t \approx O(\log(1/\eps))$ to get an additive approximation of $\Tr(Q \rho_{\beta,H})$. To prove \eqref{eq:decay-exponentially-weighted-norm} we will regard $\cK^{\beta} = \cK^{\beta}(H)$ as a perturbation of a reference generator $\cK_0$:
\[
\cK^{\beta} = \cK_0 + \cR.
\]
It is easy to verify that
\begin{equation}
\label{eq:ddtt12}
\frac{\dd}{\dd t} \|\cP_t Q\|_{a} \leq \left(\left.\frac{\dd}{\dd \tau}\right|_{\tau=0}\|e^{\tau \cK_0} \cP_t Q\|_{a}\right) + \|\cR \cP_t Q\|_{a}.
\end{equation}
Indeed, simply write, for small $\dd t$, $\|\cP_{t+\dd t} Q\|_{a} \approx \|(I+\dd t(\cK_0 + \cR)) \cP_t Q\|_{a}$, and so by the triangle inequality
\begin{equation}
\label{eq:ddtt12proof}
\frac{\|\cP_{t+\dd t} Q\|_{a} - \|\cP_t Q\|_{a}}{\dd t} \leq \frac{\|(I+\dd t \cK_0) \cP_t Q\|_{a} - \|\cP_t Q\|_{a}}{\dd t} + \|\cR \cP_t Q\|_{a} + O(\dd t).
\end{equation}
Assume the reference generator $\cK_0$ and the residual $\cR = \cK^{\beta} - \cK_0$ satisfy inequalities of the form:
\begin{equation}
\label{eq:ineq-reference}
\frac{\dd}{\dd \tau} \|e^{\tau \cK_0} X\|_{a} \leq - \lambda \|\cN X\|_{a} \quad \forall X \in \cA,
\end{equation}
and
\begin{equation}
\label{eq:ineq-residual}
\|\cR X\|_{a} \leq \eta \|\cN X\|_{a} \quad \forall X \in \cA
\end{equation}
for some $\lambda, \eta > 0$.
Then these immediately give, from \eqref{eq:ddtt12}
\begin{equation}
\label{eq:combining-decays}
\frac{\dd}{\dd t} \|\cP_t Q\|_{a} \leq -(\lambda - \eta) \|\cN \cP_t Q\|_{a}
\end{equation}
which yields the desired decay \eqref{eq:decay-exponentially-weighted-norm} provided $\eta < \lambda$.
Concretely, we take $\cK_0 = \cK^{\beta}(H_0)$ to be the generator associated to the non-interacting part of the Hamiltonian. The operators $A_i$ and $B_i$ of the pseudo-Lindbladian \eqref{eq:def-Kbeta} are chosen in such a way that we can prove the decay property \eqref{eq:ineq-reference}. Indeed, with our choice of $A_i$ and $B_i$ it turns out that $\cK^{\beta}(H_0)$ is a perturbation of the depolarizing noise generated by $-\cN$. To prove the inequality \eqref{eq:ineq-residual} on the residual $\cR$ we use a Dyson expansion in powers of $u$ to bound the terms $\tau_{\pm \beta/2}(B_i)$ in the definition \eqref{eq:def-Kbeta} of $\cK^{\beta}$.

\paragraph{Choice of pseudo-jump operators $A_i$ and $B_i$} We now completely specify the precise pseudo-Lindbladian we analyze in this paper. We take $\alpha=1/4$, $\g=N$ in \eqref{eq:def-Kbeta} and let, for $i=1,\ldots,N$:
\begin{equation}
\label{eq:pseudo-jumpA}
\begin{cases}
A_i &= \gamma_i\\
B_i &= \sech\big(\frac\beta2 \cG_0\big)(A_i)
\end{cases}
\end{equation}
where $\sech(x) = 1/\cosh(x)$ and $\cG_0 = [H_0,\cdot]$.
To motivate the choice of $B_i$'s note that \eqref{eq:def-Kbeta} can be conveniently written as:
\begin{equation}
\label{eq:def-Kbeta-2}
\cK^{\beta}(H) X = \dfrac{1}{4} \sum_{i=1}^{\g} [[A_i,X],C_i] + \{[A_i,X],S_i\}
\end{equation}
where
\begin{equation}
\begin{cases}
C^\beta_i(H) = (\tau^H_{-\beta/2}(B_i) + \tau^H_{\beta/2}(B_i))/2 = \cosh(\frac\beta2 \cG)(B_i)\\
S^\beta_i(H) = (\tau^H_{-\beta/2}(B_i) - \tau^H_{\beta/2}(B_i))/2 = -\sinh(\frac\beta2 \cG)(B_i)
\end{cases}
\end{equation}
where we let $\cG = [H,\cdot]$ so that $\tau^H_{s} = e^{s \cG}$. We would like to choose $B_i$'s so that when $H=H_0$ we get $C_i = A_i$: this will ensure that when $H=H_0$ the first term in \eqref{eq:def-Kbeta-2} will simplify to the depolarizing noise $-\cN$ (indeed it is known that, on the sub-algebra of even observables $
-\cN(X) = 
\frac{1}{4} \sum_{i=1}^N [[\gamma_i,X],\gamma_i]
$
see Remark \ref{remark:pseudo-noise-fermi}),  which will be very convenient to prove the decay condition \eqref{eq:ineq-reference}. To achieve this, it suffices to take 
\begin{equation}
\label{eq:choiceBi}
    B_i = \cosh(\beta/2 \cG_0)^{-1}(A_i).
\end{equation}
Since $\cG_0$ preserves degree (by \eqref{eq:H0degreeinvariant}), note that $\deg(B_i) = \deg(A_i) = 1$, i.e., $B_i$ is a linear combination of single-site Majoranas. With the choice of $B_i$ in \eqref{eq:choiceBi} note that $S_i = -\sinh(\frac\beta2 \cG) \cosh(\beta/2\cG_0)^{-1}(A_i)$. Importantly, we get that the pseudo-Lindbladian for $H=H_0$ simplifies to
\begin{equation}
\label{eq:KbetaH0}
\begin{aligned}
\cK^{\beta}(H_0)(X) &= \dfrac{1}{4} \sum_{i=1}^{\g} [[A_i,X],A_i] - \{[A_i,X],\tanh(\frac\beta2 \cG_0)(A_i)\}\\
&= -\cN(X) - \cD_0(X).
\end{aligned}
\end{equation}
where $\cD_0(X) = \frac14 \sum_{i=1}^{N} \{[\gamma_i,X],\tanh(\frac\beta2 \cG_0)(\gamma_i)\}$. Some of the properties of $\cD_0$ will be studied later, in Appendix \ref{app:decay2}.

\begin{remark}[Even subalgebra]
For exposition purposes, we consider in this section observables that are in the even subalgebra, i.e., that are linear combinations of even-degree Majorana monomials. This particular choice simplifies the definition of jump operators in \eqref{eq:pseudo-jumpA}. The results can be generalized to the full algebra of observables by taking $A_i$ in \eqref{eq:pseudo-jumpA} to be $\ci \Pi \gamma_i$ where $\Pi$ is the parity operator as explained in Section \ref{sec:degree-operator}. This requires special treatment of the parity operator which only complicates the exposition, so that's why we restrict our attention to the even subalgebra.
\end{remark}

\subsection{Assumption on $H_0$}

Our analysis of the propagation algorithm for the evolution $\cP_t$ will require further assumptions on $H_0$. These assumptions are primarily to control the perturbation $\cD_0$ in \eqref{eq:KbetaH0}, and the magnitude of the operators $B_i$ defined in \eqref{eq:choiceBi}.
More formally:

\begin{assumption}
\label{assumption:diagonal-hamiltonian-fermion}
Write the quadratic Majorana $H_0$ as 
\begin{equation}
    \label{eq:H0diagonal}
    H_0 = \frac{\ci}{4} \sum_{1\leq i,j \leq N} A_{ij} \gamma_i \gamma_j
\end{equation}
where $A \in \RR^{N\times N}$ is skew-symmetric, i.e., $A^T = -A$. We assume that there is a function $\F : \RR_+ \to \RR_+$ with $\F(x) \geq 1$ for all $x \geq 0$ such that for all $\beta \geq 0$
\begin{equation}
\label{eq:sechboundassumption}
\|\sech(\ci A \beta/2 )\|_{\row} \leq \F(\beta)
\end{equation}
\begin{equation}
\label{eq:tanhboundassumption}
\|\tanh(\ci A \beta/2 )\|_{\row} \leq \F(\beta)
\end{equation}
where for a $N\times N$ matrix $M$
\[
\|M\|_{\row} = \max_{i \in \{1,\ldots,N\}} \sum_{j=1}^N |M_{ij}|.
\]
\end{assumption}

Equation \eqref{eq:sechboundassumption} is used to bound the magnitude of $B_i$ in \eqref{eq:pseudo-jumpA}, and Equation \eqref{eq:tanhboundassumption} is used to bound the perturbation term $\cD_0$ in \eqref{eq:KbetaH0}.

Although the conditions above may look nonstandard, and tailored to our specific setup,
it turns out that growth functions of the form $\F(\beta)$ are well studied in the literature, and are related to the spatial decay of the Green's function for $H_0$. In fact, the analysis of the cumulant expansion in \cite{chen2025convergence} is based on very similar assumptions (see \cite[Definition 5]{chen2025convergence}). For example, it is known that for a $D$-dimensional square lattice the inequalities above are satisfied with $\F(x) = C_D (1+2Dx)^{D/2+1}$ (see Appendix \ref{app:lattice}).

\subsection{Contraction}

Our first main result shows that under appropriate assumptions, the evolution $\cP_t$ is a contraction for an exponentially-weighted norm $\|\cdot\|_{a}$ for a range of $a \in [a_0,a_0+v]$.

\begin{theorem}
\label{thm:gibbs1}
Let $H = H_0 + uV$ be a Hamiltonian satisfying Assumptions \ref{assumption:hamiltonian} and \ref{assumption:diagonal-hamiltonian-fermion}. Let $\beta \geq 0$ and $\cP_t$ be the semigroup defined by the pseudo-Lindbladian \eqref{eq:semigroup} with the choice of $A_i$ and $B_i$ given in \eqref{eq:pseudo-jumpA}.

Assume that
\begin{equation}
\label{eq:assumptionub}
u (e^{2 \delta \Delta \beta} - 1) < C\, \F(\beta)^{-(3+\delta)/2} \, e^{-2\Delta \beta}
\end{equation}
where $C = (1/32) e^{-1} 2^{-(1+\delta)/2}$. Let 
\begin{equation}
    \label{eq:defa0}
    a_0 = \frac{1}{2} \log(2\F(\beta))
\end{equation}
and $v=1/(1+\delta)$.
Then for any $a_0 \leq a \leq a_0 + v$ and all $X \in \cA$ even,
\begin{equation}
\label{eq:decay-semigroup}
\frac{\dd}{\dd t} \|\cP_t X\|_{a} \leq -\frac{1}{4} \|\cN \cP_t X\|_a.
\end{equation}
\end{theorem}
\begin{proof}
We follow the approach described in Section \ref{sec:overviewgibbs}  and focus on proving inequalities \eqref{eq:ineq-reference} and \eqref{eq:ineq-residual}. Set $f = \F(\beta)$.
On the one hand, Proposition \ref{prop:decayH01} of Appendix \ref{app:decay2} shows for all $a \geq 0$ \eqref{eq:ineq-reference} holds with
\[
\lambda_a \geq 1 - f e^{-2a}.
\]

On the other hand, Theorem \ref{thm:bound-residual-K-KM} of Appendix \ref{app:boundKK} shows that for any $X$ even and any $a \geq 0$,
\begin{equation}
\|\cR X\|_{a} \leq \eta_a \|\cN X\|_{a}
\end{equation}
with
\begin{equation}
\label{eq:etaa}
\eta_a = 4e^{2\beta \Delta + a}f\frac{\rho_a}{1-\rho_a} 
\end{equation}
assuming
\[
\rho_a = u (e^{2\beta \Delta \delta} - 1) e^{a\delta} < 1.
\]
We now show that under the assumption \eqref{eq:assumptionub} and $a \in [a_0,a_0+v]$, we have $\lambda_a - \eta_a \geq 1/4$.
First note that for $a \geq a_0 = \log(2 f)/2$  we have
\[
\lambda_a \geq 1/2.
\]
We are now going to analyze the numerator of $\eta_a$ given in \eqref{eq:etaa}. Call $x = \Delta \beta$, and write $\rho_a = \kappa e^{a\delta}$ with
\[
\kappa = u(e^{2\delta x}-1).
\]
Since $a(1+\delta) \leq (a_0+v)(1+\delta) \leq a_0(1+\delta)+1$ we have
\begin{equation}
\label{eq:keyineqs}
\begin{aligned}
4e^{2x + a} f \rho_a &= 4e^{2x+a} f \kappa e^{a\delta}\\
&= 4e^{a(1+\delta)} f \kappa e^{2x}\\
&\leq 4e^{a_0(1+\delta)+1} f \kappa e^{2x}\\
&\leq 4 e 2^{(1+\delta)/2} f^{(3+\delta)/2} \kappa e^{2x}\\
&\leq 1/8
\end{aligned}
\end{equation}
where the last inequality is precisely the assumption \eqref{eq:assumptionub}.
Since $f \geq 1, x, a \geq 0$ the inequalities above also tell us that $\rho_a \leq 1/32 \leq 1/2$, and that $\eta_a \leq \frac{1/8}{1/2} = 1/4$. Thus $\lambda_a - \eta_a \geq 1/4$ for the range of $a$ provided.
\end{proof}

\subsection{Short-step approximation}

\newcommand{\cPt}{\widetilde{\cP}}

In this section we show how for small $h$, the short-time evolution $\cP_h X$ can be well approximated by some $\widetilde{\cP}_h X$ that can be computed efficiently. Our main result is the following.

\begin{theorem}
\label{thm:efficient}
Under the assumptions of Theorem \ref{thm:gibbs1}, for any $0< h < 1$, there is a linear mapping $\widetilde{\cP}_h : \cA \to \cA$ such that the following holds for all $X \in \cA_{\leq \ell}$:
\begin{equation}
\label{eq:Ph-Ptildeh}
\|\cP_h X - \widetilde{\cP}_h X\|_{a_0} \leq c\ell^2 h^2 \|X\|_{a_0}
\end{equation}
where $c > 0$ is an absolute constant.
Furthermore $\widetilde{\cP}_h$ can be computed in time $N^{O(|\log h|)}$.
\end{theorem}
\begin{proof}
To prove this theorem, we introduce for each $M \geq 0$, the first-order Euler approximations
\begin{equation}
\label{eq:cPtgibbs}
\cPt_{h,M} X = (I + h \cK^{\beta}_{M}) X 
\end{equation}
where $\cK^{\beta}_M$ is obtained from $\cK^{\beta}$ by replacing the terms $\tau_{\pm \beta/2}(B_i)$ in \eqref{eq:def-Kbeta} by the first $M$ terms of the Dyson series (see Appendix \ref{app:complex-time}). Then
\[
\begin{aligned}
\cP_hX - \cPt_{h,M}X &= \Bigl(\cP_h X - (I+h \cK^{\beta})X\Bigr) + h(\cK^{\beta} - \cK^{\beta}_{M}) X\\
& =\int_0^h (h-s) \cP_s (\cK^{\beta})^2 X\,\dd s + h(\cK^{\beta} - \cK^{\beta}_{M}) X.
\end{aligned}
\]
Set $a=a_0$. 

Since $\cP_s$ is nonexpansive in $\|\cdot\|_{a}$ norm we have
\[
\left\|\int_0^h (h-s) \cP_s (\cK^{\beta})^2 X\,\dd s\right\|_{a} \leq \frac{h^2}{2} \|(\cK^{\beta})^2 X\|_{a} \leq \frac{h^2}{2} c \ell^2 \|X\|_{a}
\]
for $\deg X \leq \ell$, where we used the fact
\begin{equation}
\label{eq:boundK2}
\|(\cK^{\beta})^2 X\|_{a} \leq c \ell^2 \|X\|_a,
\end{equation}
which we prove below in Lemma \ref{lem:boundK2}.
Furthermore we know from Theorem \ref{thm:bound-residual-K-KM} that
\begin{equation}
\|(\cK^{\beta} - \cK^{\beta}_M)X\|_{a} \leq e^{2\beta \Delta+a} f \cdot (M+1) \frac{\rho^M}{(1-\rho)} \|\cN X\|_{a}
\leq 2^{-(M-1)} \|\cN X\|_a \leq 2^{-(M-1)} \ell \|X\|_a,
\end{equation}
where $\rho = \rho_a$ and we used the fact that $e^{2\beta \Delta + a} f \rho/(1-\rho) \leq (1/32)/(7/8) \leq 1$ (see \eqref{eq:keyineqs}) and $(M+1)\rho^M \leq (2\rho)^M \leq (1/2)^M$.
It follows that
\[
\|\cP_h X - \cPt_{h,M} X\|_{a} \leq \left( c h^2 + h 2^{-(M-1)}  \right) \ell^2  \|X\|_{a}.
\]
Note that the second term in the parenthesis scales with $h$, whereas we need the error to scale with $h^2$. Choose $M$ such that $2^{-(M-1)} \leq h$, i.e.,
\be
\label{eq:choiceM}
M = 1+\left\lceil \log_2(1/h) \right\rceil.
\ee
This gives, as desired:
\[
\|\cP_h X - \cPt_{h,M} X\|_{a} \leq (c+1) \ell^2 h^2 \|X\|_{a}.
\]
From point 4 of Lemma \ref{lem:complex-time-evolution} we know that $\cPt_{h,M}$ can be computed in time $N^{O(M)} = N^{O(|\log h|)}$.
It remains to prove the bound \eqref{eq:boundK2}.

\begin{lemma}
\label{lem:boundK2}
There is an absolute constant $c$ such that for every even observable $X$ of degree $\ell$, we have
    $\|(\cK^{\beta})^2 X\|_{a} \leq c \ell^2 \|X\|_{a}$.
\end{lemma}
\begin{proof}
We write $\cK^{\beta} = \cK_0 + (\cK^\beta- \cK_0)$ where $\cK_0 = \cK^{\beta}(H_0)$ is the reference generator.
We know from Theorem \ref{thm:bound-residual-K-KM}
\[
\|(\cK^{\beta}-\cK_0) X\|_a \leq 2 f e^{2x+a} \frac{\rho_a}{1-\rho_a} \|\cN X\|_a \leq \frac{1}{4} \|\cN X\|_a.
\]
For the reference generator $\cK_0$, from Proposition \ref{prop:boundK0} in Appendix \ref{app:decay2} we have (using $a = a_0 = 1/2 \log(2 f)$)
\[
\|\cK_0 X\|_a \leq (1+ f e^{-2a}) \|\cN X\|_a \leq 3/2 \|\cN X\|_a.
\]
Thus putting together, we get
\[
\|\cK^{\beta} X\|_a \leq 2 \|\cN X\|_a,
\]
and so, in particular $\|(\cK^{\beta})^2 X\|_a \leq 2 \|\cN \cK^{\beta} X\|_a$.

Using the same ideas as the proof of point 3 in Lemma \ref{lem:complex-time-evolution} (using the weight $\cN e^{a\cN}$ instead of $e^{a\cN}$, which adds a factor $\deg(Y_m) \leq d + m\delta$ in \eqref{eq:eaNYm} and changes the term $(\alpha)_M$ in \eqref{eq:seriestail} by $(\alpha)_{M+1}$) we can also show that
\[
\|\cN (\cK^{\beta} - \cK_0) X\|_{a} \leq 2 f e^{2x+a_0} \frac{\rho_a}{(1-\rho_a)^2} \|\cN^2 X \|_a \leq c \|\cN^2 X\|_a
\]
for some absolute constant $c$, using our assumptions. We can also show, in the same way as the proof of Proposition \ref{prop:boundK0} that
\[
\|\cN \cK_0 X\|_a \leq (1+ f e^{-2a}) \|\cN^2 X\|_a \leq (3/2) \|\cN^2 X\|_a.
\]
So putting things together we get
\[
\|(\cK^{\beta})^2 X\|_a \leq c' \|\cN^2 X\|_a \leq c'\ell^2 \|X\|_a
\]
when $\deg(X) \leq \ell$.
\end{proof}
\end{proof}

\subsection{Simulation algorithm}

To simulate the dynamics $\cP_t A$ we use the following propagation algorithm up to time $t$
\be
\label{eq:propagation-algo-weak-gibbs}
\begin{cases}
A_0 = A\\
A_{j+1} = \Trunc_{\ell} \cPt_{h,M} A_j \qquad (j=0,\ldots,m-1),
\end{cases}
\ee
where $\cPt_{h,M}$ is defined in \eqref{eq:cPtgibbs} with $M$ given by \eqref{eq:choiceM}.

\begin{corollary}
\label{cor:simulation-gibbs}
Under the assumptions of Theorem \ref{thm:gibbs1}, fix a horizon $t > 0$, an $\eps > 0$ and an even observable $A \in \cA$ of degree $d$. 
The propagation algorithm \eqref{eq:propagation-algo-weak-gibbs} with $\ell = d+O(\log((1+t)/\eps))$ and $m = t/h = \poly(t,\ell,\eps^{-1})$ computes an approximation $\widetilde{A}$ of $\cP_t A$ that satisfies
\begin{equation}
\label{eq:convergencepropgibbs}
\|\cP_t A - \widetilde{A}\|_{\Fone} \leq \eps \|A\|_{a_0} \leq \eps e^{a_0 d} \|A\|_{\Fone}
\end{equation}
in time $N^{O(d+\log((1+t)/\eps))}$.
\end{corollary}
\begin{proof}
We use Theorem \ref{thm:abstract} and Corollary \ref{cor:criterion} with $\|\cdot\| = \|\cdot\|_{a_0}$.
Note that since we have contraction of $\|\cP_t X\|_{a}$ for all $a \in [a_0,a_0 + v]$, this allows us to show that for $\deg X = d$, and $\ell \geq d$
\[
\|(I-\Trunc_{\ell}) \cP_t X\|_{a_0} \leq e^{-v(\ell-d)} \|X\|_{a_0} \quad \forall t \geq 0.
\]
So we apply Corollary \ref{cor:criterion} with $\eta(t;A) \leq Ce^{-v(\ell-d)}$ and $q(\ell) = c \ell^2$ from \eqref{eq:Ph-Ptildeh}. This gives $\|\cP_t A - \widetilde{A}\|_{a_0} \leq \eps \|A\|_{a_0} \leq \eps e^{a_0 d} \|A\|_{\Fone}$, and we use $\|\cP_t A - \widetilde{A}\|_{\Fone} \leq \|\cP_t A - \widetilde{A}\|_{a_0}$ to get the desired result \eqref{eq:convergencepropgibbs}.
\end{proof}

\subsection{Thermal observables}

\begin{proposition}[Convergence speed]
\label{prop:convergencespeed}
Under the assumptions of Theorem \ref{thm:gibbs1}, we have for any even degree-$d$ observable $X \in \cA$
\[
|\tr(\cP_t X) - \Tr(\rho_{\beta,H} X)| \leq e^{-t/4} e^{a_0 d} \|X\|_{\Fone}
\]
with $a_0$ as defined in \eqref{eq:defa0}.
\end{proposition}
\begin{proof}
Let $X \in \cA$, and let $Y_t = (I-\Trunc_0)\cP_t X = \cP_t X - \tr(\cP_t X) I$. We want to show that $\|Y_t\|_{\Fone}$ converges exponentially fast to 0. Note that (cf. Equation \eqref{eq:ddtt12})
\[
\frac{\dd}{\dd t} \|Y_t\|_{\Fone} \leq \left.\frac{\dd}{\dd \tau}\right|_{\tau=0} \|e^{\tau \cK_0} Y_t\|_{\Fone} + \|\cR Y_t\|_{\Fone}
\]
assuming $\cK^{\beta} = \cK_0 + \cR$. Indeed, one can prove the above using the same elementary method in \eqref{eq:ddtt12proof} and by using the fact that $\cK^{\beta}(I) = 0$ which implies that $Y_{t+\dd t} \approx Y_t + \dd t (I-\Trunc_0) \cK^{\beta} Y_t$, and that $\|(I-\Trunc_0)X\|_{\Fone} \leq \|X\|_{\Fone}$.

Under the assumptions of Theorem \ref{thm:gibbs1}, for any $a\in[a_0,a_0+v]$ we get
\[
\frac{\dd}{\dd t} \|Y_t\|_{a} \leq -\frac{1}{4} \|\cN Y_t\|_{a} \leq -\frac{1}{4} \|Y_t\|_a,
\]
where we used the fact that $\|\cN Y_t\|_a \geq \|Y_t\|_a$ since the constant term of $Y_t$ is zero. This implies that
\[
\|Y_t\|_{a} \leq e^{-t/4} \|Y_0\|_{a} 
\]
In particular, since $\tr(\cP_t X) - \Tr(\rho_{\beta,H} X) = -\Tr(\rho_{\beta,H} Y_t)$ we have
\[
|\tr(\cP_t X) - \Tr(\rho_{\beta,H} X)|  = |\Tr(\rho_{\beta,H} Y_t)| \leq \|Y_t\|_{\infty} \leq \|Y_t\|_{\Fone} \leq \|Y_t\|_a \leq e^{-t/4} \|X\|_{a} \leq e^{-t/4} e^{a d} \|X\|_{\Fone}.
\]
\end{proof}

Here below is pseudocode for the Majorana Propagation algorithm. The parameter $\eta$ is simply a rescaling of $\eps$ following \eqref{eq:convergencepropgibbs}, while the choice of $t$ is a direct consequence of the convergence speed in Proposition \ref{prop:convergencespeed}. The parameters $\ell,m$ follow directly from Corollary \ref{cor:criterion}, and finally the truncation constant $M$ is specified according to \eqref{eq:choiceM}.
\begin{algorithm}[H]
\caption{Thermal observables by Majorana propagation}
\label{alg:thermal-observable-propagation}
\begin{algorithmic}[1]
\Require $H=H_0+uV$, $\beta\geq0$, an even observable $X$, and $0<\varepsilon<1$.
\Ensure $|\widetilde\mu-\operatorname{Tr}(\rho_{\beta,H}X)|
 \leq\varepsilon\|X\|_{\Fone}$.
\State $d\gets\deg X$;
 $\eta\gets\tfrac{\varepsilon}{2}e^{-a_0d}$;
 $t\gets4\log(1/\eta)$.
\State $\ell\gets\left\lceil C_1\left(d+\log\frac{1+t}{\eta}\right)\right\rceil$.
\State $m\gets\left\lceil C_2(1+\ell)^2(1+t)^2/\eta\right\rceil$;
 $h\gets t/m$.
\State $M\gets1+\left\lceil\log_2(1/h)\right\rceil$.
\State Precompute the operators $B_i$ defined in \eqref{eq:pseudo-jumpA}, and the truncated terms $\tau_{\pm\beta/2}^{M}(B_i)$ that are used in the definition of $\cK_{\beta}^M$ using point 4 of Lemma \ref{lem:complex-time-evolution}.
\State $X_0\gets X$.
\For{$j=0,\ldots,m-1$}
  \State $X_{j+1}\gets \Trunc_{\ell}\bigl(X_j+h\cK_M^\beta X_j\bigr)$.
\EndFor
\State \Return $\widetilde\mu\gets\operatorname{tr}(X_m)$.
\end{algorithmic}
\end{algorithm}

The following theorem summarizes the Majorana propagation algorithm to compute thermal expectations, based on the pseudo-Lindbladian \eqref{eq:def-Kbeta}.

\begin{theorem}
\label{thm:gibbs-observables}
Under the assumptions of Theorem \ref{thm:gibbs1},
let $\eps > 0$ and $A$ an even degree-$d$ observable. Consider running the propagation algorithm \eqref{eq:propagation-algo-weak-gibbs} to time $t = 4 \log(2 e^{a_0 d}/\eps)$ and with precision $\eps' = \eps e^{-a_0 d} / 2$, and let $\widetilde{A}$ be the obtained approximation of $\cP_t A$. Let $\tilde{\mu} = \tr(\widetilde{A})$. Then $|\tilde{\mu} - \Tr(\rho_{\beta,H} A)| \leq \eps\norm{A}_1$. The complexity of the algorithm is $N^{O(\deg A + \log((1+t)/\eps))}$.
\end{theorem}
\begin{proof}
Combining Corollary \ref{cor:simulation-gibbs} and Proposition \ref{prop:convergencespeed} we immediately get:
\[
\begin{aligned}
|\Tr(\rho_{\beta,H} A) -\tilde{\mu}| &\leq |\Tr(\rho_{\beta,H} A) - \tr(\cP_t A)| + |\tr(\cP_t A) - \tr(\widetilde{A})|\\
&\leq \eps/2 \|A\|_{\Fone} + \|\cP_t A - \widetilde{A}\|_{\Fone}\\
&\leq \eps/2\|A\|_{\Fone} + \eps' e^{a_0 d} \|A\|_{\Fone} \leq \eps \|A\|_{\Fone}.
\end{aligned}
\]
\end{proof}

\paragraph{Acknowledgments} GF and HF were partially supported by UK Research and Innovation (UKRI) under the UK government's Horizon Europe funding guarantee EP/X032051/1. OF acknowledges funding by the European Research Council (ERC Grant AlgoQIP, Agreement No. 851716).

\paragraph{AI statement} OpenAI's ChatGPT Sol 5.6 and ChatGPT Astra 6 assisted with exploring and checking proof strategies and technical details. The structure and exposition of the manuscript are a result of the authors' work. The authors are fully responsible for the mathematical claims and the final manuscript.

\printbibliography

\clearpage
\newpage

\appendix

\section{Lemmas for Hamiltonian evolution}
\label{app:weak-interaction}

Recall that $\cN:\cA \to \cA$ is defined on the basis of monomials as $\cN(P) = \deg(P) P$. If we call $\P_r:\cA \to \cA$ the projection operator on the degree-$r$ homogeneous part, then $\cN = \sum_{r \geq 0} r \P_r$.

The next lemma plays a key role in Section \ref{sec:weak-interaction}. It is a degree-weighted generalization of \cite[Theorem 3.4]{facelli2026fast}.
\begin{lemma}
\label{lemma:master-commutator-bound}
Let $H = \sum_{P \in \cB} h_P P$ be a $\Delta$-sparse Hamiltonian with $\deg([H,X]) \leq \deg(X)+ \delta$ for any $X\in\cA$ and $|h_P| \leq 1$ for all $P$. Let $\cG = \ci [H,\cdot]$, and for $\phi:\ZZ \to \RR$, define
\[
\cG_{\phi} = \sum_{r,s\geq 0} \phi(s-r)\P_s \cG \P_r.
\]
Then for any $X,Y \in \cA$,
\be
\label{eq:master}
|\HSproduct{X}{\cG_{\phi} Y}| \leq 2 \Delta \|\phi\|_{\infty,\delta} \sqrt{\HSproduct{X}{\cN X}} \sqrt{\HSproduct{Y}{\cN Y}},
\ee
where $\|\phi\|_{\infty,\delta} = \max_{|k|\leq \delta} |\phi(k)|$.
\end{lemma}
\begin{proof}
Write
\[
 X=\sum_{S\in\BB}x_SS,
 \qquad
 Y=\sum_{Q\in\BB}y_QQ.
\]
Let 
\[
\mathsf E = \{(P,Q) \in \BB \times \BB : h_P \neq 0 \text{ and } [P,Q] \neq 0 \}.
\]
For every $(P,Q)\in\mathsf E$, write
\[
 PQ=f_{P,Q}S(P,Q),\qquad |f_{P,Q}|=1.
\]
We know for $(P,Q) \in \mathsf E$, $[P,Q] = 2PQ$.
Expanding the filtered commutator gives
\[
\begin{aligned}
 \langle X,\cG_\phi Y\rangle_\F
 &= \sum_{S,Q} \overline{x_S} y_Q \< S, \cG_{\phi} Q \> \\
 &= \sum_{S,Q} \overline{x_S} y_Q \phi(\deg S - \deg Q) \<S, \ci [H,Q]\>\\
 &= \sum_{P,S,Q} \overline{x_S} y_Q \phi(\deg S - \deg Q) h_P \<S, \ci [P,Q]\>\\
 &= \sum_{(P,Q) \in \mathsf E} \overline{x_{S(P,Q)}} y_Q \phi(\deg S(P,Q) - \deg Q) 2 \ci h_P f_{P,Q}.
\end{aligned}
\]
Since $|\deg S(P,Q) - \deg Q| \leq \delta$ for $(P,Q) \in \mathsf E$, and $|h_P| \leq 1$, this implies
\begin{align}
|\langle X,\cG_\phi Y\rangle_\F|
 &\leq
 2\|\varphi\|_{\infty,\delta}
 \sum_{(P,Q)\in\mathsf E}|x_{S(P,Q)}|\,|y_Q|                         \notag\\
 &\leq
 2\|\varphi\|_{\infty,\delta}
 \left(\sum_{(P,Q)\in\mathsf E}|x_{S(P,Q)}|^2\right)^{1/2}
 \left(\sum_{(P,Q)\in\mathsf E}|y_Q|^2\right)^{1/2}.
 \label{eq:edge-cs}
\end{align}
By the $\Delta$-sparsity assumption on $H$ (equation \eqref{eq:Delta-sparsity-assumption}), we have for any fixed $Q \in \BB$,

\begin{equation}
 \bigl|\{P:(P,Q)\in\mathsf E\}\bigr|
 \leq\Delta\deg(Q).
 \label{eq:input-count}
\end{equation}
Hence we have
\be
\sum_{(P,Q)\in\mathsf E}|y_Q|^2
 \leq\Delta\sum_Q\deg(Q)|y_Q|^2
 =\Delta\langle Y,\cN Y\rangle_\F.
\ee
Similarly, for a fixed output $S$, we have
\be
\label{eq:card-E(S)}
|\mathsf{E}(S)| \leq \Delta \deg(S)
\ee
where $\mathsf{E}(S) = \{(P,Q) \in \mathsf{E} : S(P,Q) = S\}$.
Indeed, note that $(P,Q) \in \mathsf{E}(S)$ implies $[P,S] \neq 0$ and so $(P,Q) \mapsto P$ is an injection from $\mathsf{E}(S)$ to $\{P : (P,S) \in \mathsf{E}\}$. Thus \eqref{eq:card-E(S)} holds by  $\Delta$-sparsity of the Hamiltonian.
Consequently,
\be
 \sum_{(P,Q)\in\mathsf E}|x_{S(P,Q)}|^2
 = \sum_{S} |x_S|^2 \sum_{(P,Q) \in \mathsf{E}(S)} 1
 \leq\Delta\sum_S\deg(S)|x_S|^2
 =\Delta\langle X,\cN X\rangle_\F.
\ee
Substituting these estimates into~\eqref{eq:edge-cs} proves
\eqref{eq:master}.
\end{proof}

We can get the estimate from \cite[Theorem 3.4]{facelli2026fast} as an immediate consequence.

\begin{lemma}
\label{lemma:frobenius-norm-growth}
Assume $H = \sum_{P} h_P P$ is a $\Delta$-sparse Hamiltonian with $\deg([H,X]) \leq \deg(X) + \delta$ for any $X\in\cA$ and $|h_P| \leq 1$ for all $P \in \cB$. Then for any $A \in \cA$, $\|[H,A]\|_\F \leq 2 \Delta \sqrt{(\deg A)(\deg A+\delta)} \|A\|_\F$.
\end{lemma}
\begin{proof}
Take $\phi(r) = 1$, $X = \cG(A)$ and $Y = A$, where $\cG = \ci [H,\cdot]$.
Note that $\HSproduct{Y}{\cN Y} \leq \deg(Y) \|Y\|_{\F}^2$, and $\HSproduct{X}{\cN X} \leq \deg(X) \|X\|_{\F}^2 \leq (\deg(A)+\delta) \|[H,A]\|_{\F}^2$. Inequality \eqref{eq:master} gives us
\[
\|[H,A]\|_{\F}^2 \leq 2 \Delta \sqrt{\deg A(\deg A + \delta)} \|[H,A]\|_{\F} \|A\|_{\F}.
\]
Dividing by $\|[H,A]\|_{\F}$ gives the result.
\end{proof}

\begin{lemma}[Euler truncation error and stability]
\label{lemma:euler-truncation-error}
Let $\cL$ be the semigroup generator from \eqref{eq:dissipative-generator}. Then for any $h > 0$, and any $X \in \cA$ with $\deg X \leq \ell$,
\be
\label{eq:euler-truncation-error}
\|e^{h\cL}X-(I+h\cL)X\|_{\F}
 \le\frac{h^2}{2}K(\ell)K(\ell+\delta)\|X\|_{\F}
\ee
where $K(r) := \lambda r + 2 \Delta \sqrt{r(r+\delta)}$.
\end{lemma}
\begin{proof}
Using Lemma \ref{lemma:frobenius-norm-growth}, we have, for any $X \in \cA$,
\be
\label{eq:bound-LX}
\|\cL X\|_{\F} \leq K(\deg X) \|X\|_{\F}.
\ee
To get inequality \eqref{eq:euler-truncation-error} first note that:
\[
e^{h\cL}X-X-h\cL X
 =\int_0^h(h-s)e^{s\cL}\cL^2X\,\dd s.
\]
By applying \eqref{eq:bound-LX} twice, and using the nonexpansiveness of $e^{s\cL}$ with respect to the Frobenius norm, we get the desired inequality \eqref{eq:euler-truncation-error}.
\end{proof}

\clearpage
\newpage
\section{Lemmas for thermal observable estimation}
\subsection{Imaginary-time evolution and Dyson expansion}\label{app:complex-time}

\begin{lemma}[Imaginary-time evolution for weakly interacting Hamiltonians]
\label{lem:complex-time-evolution}
Let $H=H_0+uV$ be a Hamiltonian satisfying Assumption \ref{assumption:hamiltonian}.
Assume furthermore that $H_0$ satisfies
\begin{equation}
\label{eq:H0additionalassumptionDelta0}
\|[H_0,X]\|_{\Fone} \leq 2\Delta \deg(X) \|X\|_{\Fone} \quad \forall X \in \cA,
\end{equation}
for some constant $\Delta > 0$.
Let $F\in\cA$.
Then, for any $s\in\RR$, we have
\begin{equation}
 \label{eq:iterated-duhamel}
 e^{sH}Fe^{-sH}
 =
 \sum_{m=0}^{\infty}Y_m(s;F),
\end{equation}
where the $Y_m(s;F)$ are defined recursively by
\[
 \begin{aligned}
 Y_0(s;F)
 &=e^{sH_0}Fe^{-sH_0},\\
 Y_m(s;F)
 &=u\int_0^s
 e^{(s-t)H_0}[V,Y_{m-1}(t;F)]e^{-(s-t)H_0}\,dt,
 \qquad m\geq1.
 \end{aligned}
\]
Furthermore, writing $d = \deg F$, the $Y_m$ satisfy the following:
\begin{enumerate}
 \item
 $Y_m(s;F)\in\cA_{\leq d+m\delta}$.

 \item For any $m \geq 0$:
 \begin{equation}
 \label{eq:boundYm}
  \|Y_m(s;F)\|_{\Fone}
 \leq
 e^{2\Delta d|s|}
 \frac{1}{m!}
 \prod_{j=1}^{m}
 \left(\frac d\delta+j-1\right)
 \bigl(u(e^{2\Delta \delta|s|}-1)\bigr)^m
 \|F\|_{\Fone}.
 \end{equation}

 \item Let $a \geq 0$ and assume $\rho := u (e^{2\Delta \delta |s|}-1) e^{a\delta} < 1$. Then for any $M \geq 0$
 \begin{equation}
 \label{eq:tail-imaginary-time-evolution}
     \|e^{a\cN}\Bigl(e^{s H} F e^{-sH} - \sum_{m=0}^{M-1} Y_m(s;F)\Bigr)\|_{\Fone} \leq e^{2 \Delta d |s| + ad} \|F\|_{\Fone} \frac{(d/\delta)_M}{M!} \frac{\rho^M}{(1-\rho)^{\max(d/\delta,1)}}
 \end{equation}
 where $(d/\delta)_M = (d/\delta)(d/\delta+1)\dots(d/\delta + M-1)$.

 \item The tuple $(Y_0,\ldots,Y_m)$ can be computed in time complexity $N^{O(m)}$.
\end{enumerate}
\end{lemma}

\begin{proof}
We use the notations
\begin{equation}
 \begin{aligned}
 \alpha_t(\cdot)&=e^{tH_0}\cdot e^{-tH_0},\qquad
 \cG(\cdot)&=[V,\cdot].
 \end{aligned}
\end{equation}
Duhamel's identity gives
\[
 e^{sH}Fe^{-sH}
 =
 \alpha_s(F)
 +
 u\int_0^s
 \alpha_{s-t}\bigl([V,e^{tH}Fe^{-tH}]\bigr)\,dt.
\]
Iterating this identity gives
Eq.~\eqref{eq:iterated-duhamel} with the stated recursion. By assumption~\eqref{eq:H0degreeinvariant}, $\alpha_t$ preserves
degree and $\cG$ increases degree by at most $\delta$.
Induction on $m$ therefore proves point 1.

For point 2, first assume $s\geq0$.
Unrolling the iteration gives, for $m\geq1$,
\begin{equation}
 \label{eq:Ymintegralsimplex}
 \begin{aligned}
 Y_m(s;F)
 =u^m\int_{0<t_1<\cdots<t_m<s}
 \alpha_{s-t_m}\cG\alpha_{t_m-t_{m-1}}\cG
 \cdots
 \cG\alpha_{t_2-t_1}\cG\alpha_{t_1}(F)
 \,\dd t_1\cdots\dd t_m.
 \end{aligned}
\end{equation}
Lemma~\ref{lemma:inequalities} and
Eq.~\eqref{eq:Delta-sparsity-assumption} give, for every $X\in\cA$,
\begin{align}
 \|\cG(X)\|_{\Fone}
 \leq2\Delta\deg(X)\|X\|_{\Fone},\qquad
 \|\alpha_s(X)\|_{\Fone}
 \leq e^{2\Delta\deg(X)|s|}\|X\|_{\Fone}.
 \label{eq:alphasXbound}
\end{align}
To obtain Eq.~\eqref{eq:alphasXbound}, write
$\cG_0=[H_0,\cdot]$ and use
\[
 \alpha_s=\sum_{j=0}^{\infty}\frac{s^j}{j!}\cG_0^j.
\]
Since $\cG_0$ preserves degree,
Eq.~\eqref{eq:H0additionalassumptionDelta0} implies
\[
 \|\cG_0^jX\|_{\Fone}
 \leq(2\Delta\deg(X))^j\|X\|_{\Fone},
\]
which gives the stated exponential bound.

With $d=\deg F$, the integrand of
Eq.~\eqref{eq:Ymintegralsimplex} satisfies
\[
 \begin{aligned}
 &\|\alpha_{s-t_m}\cG\alpha_{t_m-t_{m-1}}\cG
       \cdots\cG\alpha_{t_1}(F)\|_{\Fone}\\
 &\quad\leq
 (2\Delta)^m
 \prod_{j=1}^{m}\bigl(d+(j-1)\delta\bigr)
 \|F\|_{\Fone}
 \times
 \exp\Bigl(
 2\Delta(s-t_m)(d+m\delta)
 +2\Delta(t_m-t_{m-1})(d+(m-1)\delta)
 +\cdots+2\Delta t_1d
 \Bigr)\\
 &\quad=
 (2\Delta)^m e^{2\Delta ds}
 \prod_{j=1}^{m}\bigl(d+(j-1)\delta\bigr)
 \prod_{j=1}^{m}e^{2\Delta \delta(s-t_j)}
 \|F\|_{\Fone}.
 \end{aligned}
\]
We keep the degree-growth product explicitly:
\[
 \prod_{j=1}^{m}\bigl(d+(j-1)\delta\bigr)
 =
 \delta^m
 \prod_{j=1}^{m}
 \left(\frac d\delta+j-1\right).
\]
The scalar exponential factor is symmetric in the insertion
times. Its integral over the ordered simplex is therefore
\[
 \begin{aligned}
 \int_{0<t_1<\cdots<t_m<s}
 \prod_{j=1}^{m}e^{2\Delta \delta(s-t_j)}
 \,\dd t_1\cdots\dd t_m
 =
 \frac1{m!}
 \left(\int_0^s e^{2\Delta \delta(s-t)}\,\dd t\right)^m
 =
 \frac1{m!}
 \left(\frac{e^{2\Delta \delta s}-1}{2\Delta\delta}\right)^m.
 \end{aligned}
\]
Consequently,
\[
 \begin{aligned}
 \|Y_m(s;F)\|_{\Fone}
 &\leq
 u^m(2\Delta)^m e^{2\Delta ds}
 \delta^m
 \prod_{j=1}^{m}
 \left(\frac d\delta+j-1\right)
 \frac1{m!}
 \left(\frac{e^{2\Delta \delta s}-1}{2\Delta \delta}\right)^m
 \|F\|_{\Fone}\\
 &=
 e^{2\Delta ds}
 \frac{1}{m!}
 \prod_{j=1}^{m}
 \left(\frac d\delta+j-1\right)
 \bigl(u (e^{2\Delta \delta s}-1)\bigr)^m
 \|F\|_{\Fone}
 \end{aligned}
\]
as desired.
For $s<0$, apply the same argument at time $|s|$ to
$-H_0$ and $-V$.

For point 3, we simply use the fact that
\begin{equation}
\label{eq:eaNYm}
\|e^{a\cN} Y_m\|_{\Fone} \leq e^{a(d+m\delta)} \|Y_m\|_{\Fone}
\end{equation}
since $\deg(Y_m) \leq d+m\delta$, and apply the bound \eqref{eq:boundYm}. Then use the inequality
\begin{equation}
\label{eq:seriestail}
\sum_{m=M}^{\infty} \frac{(\alpha)_{m}}{m!} r^m \leq \frac{(\alpha)_M}{M!} \frac{r^M}{(1-r)^{\max(\alpha,1)}}
\end{equation}
where $(\alpha)_m = \alpha (\alpha+1)\dots (\alpha + m-1)$ when $r < 1$. We demonstrate this in Remark~\ref{rem:pochhammer-sum}.

Point 4: Differentiating the recursion shows that
$(Y_0,\ldots,Y_m)$ satisfies the following set of differential equations:
\begin{equation}
 \begin{aligned}
 Y_0'(s)&=[H_0,Y_0(s)],\\
 Y_j'(s)&=[H_0,Y_j(s)]+u[V,Y_{j-1}(s)],
 \qquad j=1,\ldots,m,
 \end{aligned}
\end{equation}
with initial condition
\[
 (Y_0(0),\ldots,Y_m(0))=(F,0,\ldots,0).
\]
By point 1, the $j$th component belongs to
$\cA_{\leq d+j\delta}$.
Expanding each component of $(Y_0,\ldots,Y_m)$ in the basis $\cB$, we get a coupled system of linear ODEs in at most $(m+1)\dim\cA_{\leq d+m\delta} = (m+1)N^{O(m)}$ variables. Thus a standard integration scheme for ODEs can be used to compute $(Y_0(s),\ldots,Y_m(s))$ in time $N^{O(m)}$.
\end{proof}
\begin{remark}\label{rem:pochhammer-sum}
In this remark we prove \eqref{eq:seriestail}.
For \(\alpha>0\), \(M\geq0\), and \(0\leq r<1\), we have
\begin{equation}
\label{eq:rising-factorial-tail}
\sum_{m=M}^{\infty}
\frac{(\alpha)_m}{m!}r^m
\leq
\frac{(\alpha)_M}{M!}
\frac{r^M}{(1-r)^{\max\{\alpha,1\}}}.
\end{equation}
We prove this by considering the cases \(\alpha\leq1\) and
\(\alpha\geq1\). If \(0<\alpha\leq1\), then
\[
\frac{(\alpha)_{m+1}/(m+1)!}{(\alpha)_m/m!}
=
\frac{\alpha+m}{m+1}
\leq1.
\]
Thus the sequence \((\alpha)_m/m!\) is nonincreasing, and hence
\[
\sum_{m=M}^{\infty}
\frac{(\alpha)_m}{m!}r^m
\leq
\frac{(\alpha)_M}{M!}
\sum_{m=M}^{\infty}r^m
=
\frac{(\alpha)_M}{M!}
\frac{r^M}{1-r}.
\]
Now suppose \(\alpha\geq1\). Writing \(m=M+k\), we have
\[
\frac{(\alpha)_{M+k}}{(M+k)!}
=
\frac{(\alpha)_M}{M!}
\frac{(\alpha+M)_k}{(M+1)_k}.
\]
For every \(j\geq0\),
\[
\frac{\alpha+M+j}{M+1+j}
\leq
\frac{\alpha+j}{1+j}
\]
Taking the product over \(j=0,\ldots,k-1\) therefore gives
\[
\frac{(\alpha+M)_k}{(M+1)_k}
\leq
\frac{(\alpha)_k}{k!}.
\]
Consequently,
\[
\sum_{m=M}^{\infty}
\frac{(\alpha)_m}{m!}r^m
=
\frac{(\alpha)_M}{M!}r^M
\sum_{k=0}^{\infty}
\frac{(\alpha+M)_k}{(M+1)_k}r^k\leq
\frac{(\alpha)_M}{M!}r^M
\sum_{k=0}^{\infty}
\frac{(\alpha)_k}{k!}r^k=
\frac{(\alpha)_M}{M!}
\frac{r^M}{(1-r)^\alpha},
\]
where in the last step we used the binomial-series identity
\[
\sum_{k=0}^{\infty}\frac{(\alpha)_k}{k!}r^k
=
(1-r)^{-\alpha}.
\]
Combining the two cases proves Eq.~\eqref{eq:rising-factorial-tail}.
\end{remark}

\subsection{Bounds on the pseudo-Lindbladian generator}
\label{app:boundKK}
Consider the generator $\cK^{\beta} = \cK^{\beta}(H)$ defined in \eqref{eq:def-Kbeta} for some operators $A_1,\ldots,A_{\g},B_1,\ldots,B_{\g}$. Given $M \geq 0$, we define $\cK^{\beta}_{M}$ to be the generator obtained from $\cK^{\beta}$ 
by replacing the terms $\tau_{\pm \beta/2}(B_i)$ by the first terms $m=0$ to $m=M-1$ of the Dyson series (see Equation \eqref{eq:iterated-duhamel}), i.e.,
\begin{equation}
    \label{eq:truncatedKbeta}
\cK^{\beta}_M X = \alpha \sum_{i=1}^{\g} \left([A_i,X] \widetilde{\tau}_{-\beta/2,M}(B_i) - \widetilde{\tau}_{\beta/2,M}(B_i)[A_i,X]\right)
\end{equation}
where $\widetilde{\tau}_{s,M}(F) = \sum_{m=0}^{M-1} Y_m(s;F)$ in the notation of \eqref{eq:iterated-duhamel}.
The goal of this appendix is to bound the difference $\|(\cK^{\beta} - \cK^{\beta}_M) X\|_{a}$ for weighted norms.
For $M=0$, this is simply a bound on $\cK^{\beta} X$ which is useful in various proofs (in particular for Theorem \ref{thm:efficient}). For $M=1$, note that $\cK^{\beta}_{M} = \cK^{\beta}(H_0)$ is the generator of the noninteracting part of the Hamiltonian, and a bound on the magnitude of $\cK^{\beta}(H) - \cK^{\beta}(H_0)$ is useful to prove \eqref{eq:ineq-residual}. More generally, bounding $\cK^{\beta} - \cK^{\beta}_M$ is useful when constructing efficient approximations of $\cK^{\beta}$. Our main theorem is the following.

\begin{theorem}
\label{thm:bound-residual-K-KM}
Consider a Hamiltonian $H=H_0+uV$ which satisfies Assumptions \ref{assumption:hamiltonian} and \ref{assumption:diagonal-hamiltonian-fermion}.
Let $\cK^{\beta}$ be the generator \eqref{eq:def-Kbeta} where the pseudo-jump operators $A_i$ and $B_i$ are specified by \eqref{eq:pseudo-jumpA} and let $\cK^{\beta}_M$ be its truncation as defined in \eqref{eq:truncatedKbeta}. Let $a \geq 0$ and assume that
\begin{equation}
\label{eq:rhoaassumption}
\rho_a := u  (e^{2\Delta \delta \beta}-1)e^{a\delta} < 1.
\end{equation}
Let $X \in \cA$ be even. Then
\begin{equation}
\label{eq:tailKbeta-KbetaM}
\| (\cK^{\beta} - \cK^{\beta}_M) X \|_{a} \leq \F(\beta) e^{2 \Delta \beta + a} (M+1) \frac{\rho_a^M}{(1-\rho_a)} \|\cN X\|_{a}.
\end{equation}
\end{theorem}
\begin{proof}
Assume that $X = \gamma_S \in \cB$ is a Majorana string of even degree, and let $i \in \{1,\ldots,N\}$. If $i \notin S$ then $[\gamma_i,X] = 0$, otherwise $[\gamma_i,X]$ is, up to a phase, $2\gamma_{S\setminus \{i\}}$. It follows that
\begin{equation}
\label{eq:K-KM123}
\cK^{\beta}(X) - \cK^{\beta}_M(X) = \frac14 \sum_{i \in S} [\gamma_i,X] (\tau_{-\beta/2}(B_i) - \widetilde{\tau}_{-\beta/2,M}(B_i)) -  (\tau_{\beta/2}(B_i) - \widetilde{\tau}_{\beta/2,M}(B_i)) [\gamma_i,X].
\end{equation}
Let
\[
D = \max_{i \in \{1,\ldots,N\}} \|\Bigl(\tau_{\pm \beta/2}(B_i) - \widetilde{\tau}_{\pm \beta/2,M}(B_i)\Bigr)\|_{a}.
\]
Then by noting that when $i \in S$,  $\| [\gamma_i,X]\|_{a} = 2 e^{a(\deg(X)-1)} \leq 2 e^{a \deg(X)}$ we see that, using submultiplicativity of $\|\cdot\|_{a}$ (Lemma \ref{wth:lem:weighted-product}) and the fact that there are $|S|=\deg(X)$ terms in \eqref{eq:K-KM123}
\[
\|(\cK^{\beta} - \cK^{\beta}_M)X\|_{a} \leq \deg(X) e^{a \deg X} D = D \|\cN X\|_{a}.
\]
Using the triangle inequality we immediately obtain the inequality above for an arbitrary even observable $X$:
\[
\|(\cK^{\beta} - \cK^{\beta}_M)X\|_a \leq D \|\cN X\|_a.
\]

It remains to bound $D$. This is precisely the object of point 3 of Lemma \ref{lem:complex-time-evolution}. Note that $d=\deg(B_i)\leq 1$, and so $(d/\delta)_M \leq (M+1)!$. Furthermore, $\|B_i\|_{\Fone} \leq \F(\beta)$ by assumption \eqref{eq:sechboundassumption} (see Remark \ref{rem:boundBj}). This gives the desired bound \eqref{eq:tailKbeta-KbetaM}, with assumption \eqref{eq:rhoaassumption} coming from the assumption of point 3 in Lemma \ref{lem:complex-time-evolution}.
\end{proof}

\subsection{Decay lemma}
\label{app:decay2}

The goal of this section is to study the pseudo-Lindbladian $\cK^{\beta}(H)$ from Section \ref{sec:pseudo-lindbladian} and its evolution when $H=H_0$. In particular, we prove the decay inequality \eqref{eq:ineq-reference} which is crucial for the proof of Theorem \ref{thm:gibbs1}.

Consider a quadratic Majorana Hamiltonian 
\begin{equation}
\label{eq:app:H0}
H_0 = \frac{\ci}{4} \sum_{ij} A_{ij} \gamma_i \gamma_j
\end{equation}
where $A=-A^T$ is real. Call $\cG_0 = [H_0,\cdot]$. The pseudo-Lindbladian we study here is \eqref{eq:KbetaH0}, i.e.,
\begin{equation}
    \label{eq:K0app}
\cK_0(X) = -\cN(X) - \cD_0(X)
\end{equation}
where
\[
\cD_0(X) = \frac{1}{4} \sum_{j=1}^N \{[\gamma_j,X],\tanh\left(\frac\beta2 \cG_0\right)(\gamma_j)\}.
\]
Since $\cG_0$ preserves degree, note that $\cG_0(\cA_1) \subseteq \cA_1$. Since $\cA_1$ is spanned by $\{\gamma_1,\dots,\gamma_N\}$, we can represent $\cG_0$ as a $N\times N$ matrix in this basis. More precisely, if we call $L$ this matrix, then we have:
\begin{equation}
\label{eq:matrixL}
[H_0,\gamma_j] = \sum_{i=1}^N L_{ij} \gamma_i.
\end{equation}
It is easy to check that, for $H_0$ given by \eqref{eq:app:H0}, we have $L = \ci A$. Note that $L$ is a Hermitian, pure imaginary matrix.
Then clearly $\tanh(s \cG_0)(\gamma_j) = \sum_{i=1}^N \tanh(s L)_{ij} \gamma_i$. Call \[
G = \tanh\left(\frac\beta2 L\right)
\]
and note that $G$ is a $N\times N$ Hermitian, pure imaginary, matrix. Then we have arrived at
\[
\cD_0(X) = \frac{1}{4} \sum_{i,j=1}^N G_{ij} \{[\gamma_j,X],\gamma_i\}.
\]
Let's study the action on \emph{even} Majorana strings. If $X = \gamma_S$ and $|S|$ is even then one can verify that
\[
\{[\gamma_j,\gamma_S],\gamma_i\} = \begin{cases}
    4 \phi^S_{i,j} \gamma_{S \setminus \{i,j\}} & \text {if } i \neq j \in S\\
    0 & \text{else}.
\end{cases}
\]
with $\phi^S_{i,j}\in \CC$, $\abs{\phi^S_{i,j}}=1$. We now study the decay rate of $\cK_0$ for a weighted Fourier-1 norm.
\begin{proposition}
\label{prop:decayH01}
Let $\|\cdot\|_{a} = \|e^{a\cN} \cdot\|_{\Fone}$. Consider the pseudo-Lindbladian $\cK_0$ in \eqref{eq:K0app}. For any $X$ in the even algebra,
\[
\lim_{t\to0} \frac{\|(I+t\cK_0)(X)\|_{a} - \|X\|_{a}}{t} \leq (-1 + R e^{-2a}) \|\cN X\|_{a}
\]
where
\begin{equation}
\label{eq:defR}
R = \|G\|_{\row} = \max_{i=1,\ldots,N} \sum_{j=1}^N |G_{ij}|
\end{equation}
\end{proposition}
\begin{proof}
By the considerations above, we know that if $X=\gamma_S$ is a Majorana monomial, then
\[
(I+t\cK_0)(\gamma_S) = (1-t|S|)\gamma_S - t\sum_{i,j \in S} G_{ij} \phi^S_{i,j} \gamma_{S\setminus \{i,j\}}
\]
where $\abs{\phi^S_{i,j}}=1$, and thus
\begin{equation}
\label{eq:decay00}
\begin{aligned}
\|(I+t\cK_0)(\gamma_S)\|_{a} &\leq e^{a|S|} (1-t|S|)  +  t e^{a(|S|-2)} \sum_{i,j \in S} |G_{ij}|\\
&\leq e^{a|S|} (1-t|S|) +  t e^{a(|S|-2)} R |S|
\end{aligned}
\end{equation}
where we assumed $t$ small enough so that $(1-t\abs{S}) \geq 0$ and used the fact that
\[
\sum_{i,j \in S} |G_{ij}| \leq |S| R
\]
by definition of $R$.
Now if $X$ is an arbitrary even observable, by expanding $X = \sum_{S} x_S \gamma_S$ in the basis of Majorana strings, we immediately get from \eqref{eq:decay00}
\[
\begin{aligned}
\|(I+t\cK_0)(X)\|_{a} &\leq \sum_{S} |x_S| e^{a|S|} (1-t|S|+te^{-2a} |S| R)\\
&= \|X\|_{a} + t(-1 + e^{-2a}R) \|\cN X\|_{a}.
\end{aligned}
\]
This completes the proof.
\end{proof}

The following proposition will also be useful. The proof is very similar to the proposition above, so we omit it.
\begin{proposition}
\label{prop:boundK0}
Consider the pseudo-Lindbladian $\cK_0$ defined in \eqref{eq:K0app}. Then for any $a \geq 0$, $\|\cK_0 X\|_a \leq (1+Re^{-2a}) \|\cN X\|_a$ where $R$ is as defined in \eqref{eq:defR}.
\end{proposition}

\begin{remark}[Bound on $\|B_j\|_{\Fone}$]
\label{rem:boundBj}
From the considerations above (in particular \eqref{eq:matrixL}), we see that the pseudo-jump operators $B_j$ from \eqref{eq:pseudo-jumpA} can be explicitly computed as
\[
B_j = \sum_{i=1}^N \sech(\ci A\beta/2)_{ij} \gamma_i
\]
where $A \in \RR^{N\times N}$ is the matrix that defines the quadratic Hamiltonian $H_0$. As such, we see that
\begin{equation}
    \label{eq:boundBj}
    \|B_j\|_{\Fone} \leq \sum_{i=1}^N |\sech( \ci A \beta/2)|_{ij} \leq \|\sech( \ci A \beta/2)\|_{\row}.
\end{equation}
In particular, under Assumption \ref{assumption:diagonal-hamiltonian-fermion}, we have $\|B_j\|_{\Fone} \leq \F(\beta)$.
\end{remark}

\subsection{Bound on $\F(\beta)$ for $H_0$ on a $D$-dimensional lattice}
\label{app:lattice}

For a matrix $M$, write
\[
    \|M\|_{\row}:=\max_i\sum_j|M_{ij}|.
\]

\begin{lemma}[Row-sum bounds from Chebyshev expansion]
\label{lem:chebyshev-row-bound}
Let $L$ be a $N\times N$ Hermitian matrix with
$\|L\|_{\mathrm{op}}\le J$, where $J>0$.
Let $G$ be its support graph, and define $V(k) = \sup_{i \in [N]} |B_G(i,k)|$ where $B_G(i,k)$ is the ball in $G$ of radius $k$ centered at $i$.
Suppose $f:[-J,J]\to\mathbb \RR$ has the Chebyshev expansion \begin{equation}
\label{eq:chebf}
    f(x)=\sum_{k\ge0}a_kT_k(x/J).
\end{equation}
Then
\begin{equation}
    \|f(L)\|_{\row}
    \le \sum_{k\ge0}|a_k|\sqrt{V(k)}.
    \label{eq:chebyshev-row-bound}
\end{equation}
\end{lemma}

\begin{proof}
Since $T_k$ is a polynomial of degree $k$, the $i$th row of
$T_k(L/J)$ is supported on $B_G(i,k)$. Moreover, the spectral
theorem and the bound $|T_k(x)|\le1$ for $x\in[-1,1]$ give $\|T_k(L/J)\|_{\mathrm{op}}\le1$.
Consequently, for any $i\in [N]$, by Cauchy--Schwarz we have
\[
    \sum_j |[T_k(L/J)]_{ij}|
    \le
    \sqrt{|B_G(i,k)|}
    \left(\sum_j|[T_k(L/J)]_{ij}|^2\right)^{1/2}
    \le \sqrt{V(k)}.
\]
Thus $\|T_k(L/J)\|_{\row}\le\sqrt{V(k)}$.
The triangle inequality
therefore proves \eqref{eq:chebyshev-row-bound}.
\end{proof}

\begin{corollary}
\label{cor:lattice-thermal-row-bounds}
Consider a quadratic Majorana Hamiltonian $H_0 = \ci/4 \sum_{1\leq i,j \leq N} A_{ij} \gamma_i \gamma_j$ where $A = -A^T$ is real, and $|A_{ij}| \leq 1$ for all $i,j$. If $A$ is supported on a $D$-dimensional square lattice, then $H_0$ satisfies Assumption \ref{assumption:diagonal-hamiltonian-fermion} and $\abs{\{j\neq i : A_{ij}\neq 0 \}} \leq 2D$ and $\F(\beta) \leq C_D (1+2D\beta)^{D/2+1}$ where $C_D$ only depends on $D$.
\end{corollary}
\begin{proof}
We need to bound $\|f(\ci A \beta )\|_{\text{row}}$ for both functions $f(x)=\tanh(x/2)$, and $f(x)=\sech(x/2)$. One can show that both functions have a Chebyshev expansion \eqref{eq:chebf} on $[-J,J]$, for any $J>0$, that satisfies
\[
    |a_0|\leq 1,
    \qquad
    |a_k|\leq 2\sqrt{2}\,e^{-\alpha k}
    \quad (k\geq 1),
    \qquad
    \alpha=\operatorname{arsinh}
        \left(\frac{\pi}{2 J}\right).
\]
Indeed, both functions are analytic on and inside the Bernstein
ellipse with parameter $e^\alpha$, where their absolute values
are bounded by $\sqrt{2}$. The estimate therefore follows from
the standard bound on Chebyshev coefficients \cite{trefethen2019approximation}.
For a $D$-dimensional lattice, we have $V(k) \leq C_D(1+k)^D$ for some constant $C_D>0$. Furthermore $\|\ci \beta A\|_{\mathrm{op}}\leq 2D\beta$.
Applying Lemma~\ref{lem:chebyshev-row-bound} with $J=2D\beta$, we obtain for both functions
\[
    \|f(\ci A \beta)\|_{\row}
    \leq
    1+2\sqrt{2C_{D}}
       \sum_{k\geq 1}(1+k)^{D/2}e^{-\alpha k}
    \leq
    C'_D(1+2D\beta)^{D/2+1}.
\]
\end{proof}

\end{document}